\documentclass[journal]{IEEEtran}
\IEEEoverridecommandlockouts

\usepackage{cite}
\usepackage{amsmath,amssymb,amsfonts}
\usepackage{algorithmic}
\usepackage{graphicx}
\usepackage{textcomp}
\usepackage{xcolor}

\usepackage{orcidlink}

\usepackage{glossaries}
\glsdisablehyper    

\usepackage{amsmath, amssymb, bm} 
\usepackage{amsthm}
\usepackage{dsfont}
\usepackage{dashbox}

\usepackage{tikz}
\usetikzlibrary{calc}
\usetikzlibrary{angles, quotes}
\usetikzlibrary{decorations.pathreplacing,decorations.markings,shapes.geometric}

\usepackage{hyperref}

\usepackage{siunitx}

\usepackage[normalem]{ulem}
\usepackage{cancel}

\usepackage{makecell}
\usepackage{tabularx}

\newcommand{\C}{\mathds{C}}
\newcommand{\R}{\mathds{R}}
\newcommand{\Z}{\mathds{Z}}
\newcommand{\E}{\mathds{E}}

\newcommand{\I}[1]{\boldsymbol{I}_{#1}}

\newcommand{\Likelihood}{\mathcal{L}}
\newcommand{\CN}{\mathcal{CN}}
\newcommand{\Compl}{\mathcal{O}}
\newcommand{\U}{\mathcal{U}}

\newcommand{\frob}[1]{\left\lVert #1 \right\rVert_{\mathrm{F}}}
\newcommand{\norm}[1]{\left\lVert #1 \right\rVert_{\mathrm{2}}}
\newcommand{\abs}[1]{\left\lvert #1 \right\rvert}

\DeclareMathOperator*{\argmax}{argmax}

\newcommand{\ExpData}[1]{\E_{\dataMat} \left\{ #1 \right\}}
\newcommand{\ExpChannel}[1]{\E_{\channelcoeff} \left\{ #1 \right\}}

\newcommand{\jc}{\mathrm{j}}    

\newcommand{\UEpos}{\boldsymbol{x}_{\mathrm{s}}}
\newcommand{\RXpos}[1]{\boldsymbol{x}_{#1}}
\newcommand{\SRXaperture}{\theta_{\mathrm{srx}}}
\newcommand{\SRXradius}{R_{\mathrm{srx}}}
\newcommand{\Sradius}{R_{\mathrm{scene}}}

\newcommand{\numRX}{N}      
\newcommand{\numF}{Q}       
\newcommand{\numP}{P}       
\newcommand{\numD}{D}       
\newcommand{\numPD}{L}      
\newcommand{\numCP}{L_{\mathrm{cp}}}      
\newcommand{\RXindex}{n}
\newcommand{\Findex}{q}     
\newcommand{\Pindex}{p}   
\newcommand{\Dindex}{d} 
\newcommand{\PDindex}{l} 
\newcommand{\BW}{B}         
\newcommand{\T}{T}          

\newcommand{\Fspacing}{\Delta_{\mathrm{f}}}       
\newcommand{\Tspacing}{\Delta_{\mathrm{t}}}       

\newcommand{\carrierF}{f_{\mathrm{s}}}              
\newcommand{\carrierWl}{\lambda_{\mathrm{s}}}       
\newcommand{\carrierWn}{k_{\mathrm{s}}}             

\newcommand{\constSet}{\mathcal{C}}                 
\newcommand{\constMap}{\nu}                         
\newcommand{\constSizeMax}{M_{\mathrm{max}}}        
\newcommand{\constSize}{M}                          
\newcommand{\constSizeqd}{M_{\Findex\Dindex}}       
\newcommand{\constSymbol}{s}                        
\newcommand{\constEnergy}{\mathcal{E}_{\constSize}} 
\newcommand{\constEnergyqd}{\mathcal{E}_{\constSizeqd}}     
\newcommand{\constVar}{\sigma_{\mathrm{s}}^2}       

\newcommand{\pilotMat}{\boldsymbol{S}_{\mathrm{P}}}                 
\newcommand{\dataMat}{\boldsymbol{S}_{\mathrm{D}}}                  
\newcommand{\symbolMat}{\boldsymbol{S}}                             
\newcommand{\dataMatEst}{\widehat{\boldsymbol{S}}_{\mathrm{D}}}     
\newcommand{\dataTensEst}{\widehat{\boldsymbol{\mathcal{S}}}_{\mathrm{D}}}
\newcommand{\pilotObs}{\boldsymbol{\mathcal{Y}}_{\mathrm{P}}}       
\newcommand{\dataObs}{\boldsymbol{\mathcal{Y}}_{\mathrm{D}}}        
\newcommand{\Obs}{\boldsymbol{\mathcal{Y}}}                         
\newcommand{\pilotEnergy}{\mathcal{P}_{\mathrm{P}}}                 
\newcommand{\pilotObsEq}{\boldsymbol{Y}_{\mathrm{P}}}               
\newcommand{\ObsEq}{\boldsymbol{Y}}                                 
\newcommand{\ObsEqDD}{\boldsymbol{Y}_{\mathrm{DD}}}                 

\newcommand{\channelMat}{\boldsymbol{H}}            
\newcommand{\channel}{\boldsymbol{\beta}}           
\newcommand{\channelcoeff}{\boldsymbol{\gamma}}     
\newcommand{\channelcoeffn}{\gamma_{\RXindex}}      
\newcommand{\stMat}{\boldsymbol{A}}                 
\newcommand{\channelVar}{\sigma_{\gamma}^{2}}      
\newcommand{\channelVarN}{\sigma_{\gamma}^{2\numRX}}    
\newcommand{\channelcoeffEstPilots}{\widehat{\boldsymbol{\gamma}}_{\mathrm{p}}}
\newcommand{\channelcoeffTest}{\widetilde{\boldsymbol{\gamma}}}
\newcommand{\channelConstruct}{\widehat{\boldsymbol{H}}_{\mathrm{p}}}

\newcommand{\channelPhase}{\phi}

\newcommand{\pilotAWGN}{\boldsymbol{\mathcal{N}}_{\mathrm{P}}}       
\newcommand{\dataAWGN}{\boldsymbol{\mathcal{N}}_{\mathrm{D}}}        
\newcommand{\noiseVar}{\sigma_{\mathrm{n}}^{2}}
\newcommand{\noiseVarK}{\sigma_{\mathrm{n}}^{2K}}
\newcommand{\snr}{\mathrm{SNR}}

\newcommand{\UEposTest}{\widetilde{\boldsymbol{x}}_{\mathrm{s}}}
\newcommand{\UEposEst}{\widehat{\boldsymbol{x}}_{\mathrm{s}}}
\newcommand{\AF}{\mathrm{AF}}
\newcommand{\rmse}{\mathrm{RMSE}}
\newcommand{\ser}{\mathrm{SER}}

\newcommand{\rangeRes}{R_{\mathrm{res}}}

\newcommand{\MMLo}{\mathrm{MML}}
\newcommand{\MMLa}{\mathrm{MML_{a}}}
\newcommand{\MMLfast}{\mathrm{MML_{fast}}}
\newcommand{\Pmethod}{\textsc{Pilot}}
\newcommand{\PDmethod}{\textsc{Genie}}

\newcommand{\DDcentrmethod}{\mathrm{DD_{centr}}}
\newcommand{\DDdistrmethod}{\mathrm{DD_{distr}}}

\newcommand{\HDDcentrmethod}{\mathrm{HDD_{centr}}}
\newcommand{\HDDdistrmethod}{\mathrm{HDD_{distr}}}
\newcommand{\SDDcentrmethod}{\mathrm{SDD_{centr}}}

\newcommand{\Ngrid}{N_{\mathrm{g}}}
\newcommand{\Nmc}{N_{\mathrm{mc}}}

\newcommand{\CorrP}{\Phi^{\mathrm{P}}}
\newcommand{\corrP}{\phi^{\mathrm{P}}}
\newcommand{\CorrDD}{\Phi^{\mathrm{DD}}}
\newcommand{\corrDD}{\phi^{\mathrm{DD}}}

\newacronym{isac}{ISAC}{Integrating Sensing And Communications}
\newacronym{ff}{FF}{Far-Field}
\newacronym{nf}{NF}{Near-Field}
\newacronym{5g-nr}{5G-NR}{Fifth Generation New Radio}
\newacronym{prs}{PRS}{Positioning Reference Signal}
\newacronym{6g}{6G}{Sixth Generation}
\newacronym{pcs}{PCS}{probabilistic constellation shaping}
\newacronym{aoa}{AoA}{Angle of Arrival}
\newacronym{ula}{ULA}{Uniform Linear Array}
\newacronym{cp}{CP}{cyclic prefix}

\newacronym{ofdm}{OFDM}{Orthogonal Frequency-Division Multiplexing}
\newacronym{qam}{QAM}{Quadrature Amplitude Modulation}
\newacronym{bpsk}{BPSK}{Binary Phase-Shift Keying}
\newacronym{psk}{PSK}{Phase-Shift Keying}
\newacronym{gfsk}{GFSK}{Gaussian Frequency Shift Keying}

\newacronym{ue}{UE}{User Equipment}
\newacronym{tx}{TX}{transmitter}
\newacronym{rx}{RX}{receiver}
\newacronym{srx}{SRX}{sensing receiver}
\newacronym{das}{DAS}{Distributed Antenna System}
\newacronym{cpu}{CPU}{Central Processing Unit}
\newacronym{uca}{UCA}{Uniform Circular Array}

\newacronym{los}{LoS}{Line of Sight}

\newacronym{awgn}{AWGN}{Additive White Gaussian Noise}
\newacronym{snr}{SNR}{Signal-to-Noise Ratio}

\newacronym{mml}{MML}{Marginal Maximum Likelihood}
\newacronym{iml}{IML}{Integrated Maximum Likelihood}
\newacronym{ml}{ML}{Maximum Likelihood}
\newacronym{jml}{JML}{Joint Maximum Likelihood}
\newacronym{np}{NP}{Nuisance Parameter}
\newacronym[longplural=Parameters of Interest]{poi}{PoI}{Parameter of Interest}
\newacronym{iid}{i.i.d.}{independent and identically distributed}
\newacronym{pdf}{p.d.f.}{probability distribution function}
\newacronym{dd}{DD}{Decision-Directed}
\newacronym{sdd}{SDD}{Soft-Decision-Directed}
\newacronym{hdd}{HDD}{Hard-Decision-Directed}
\newacronym{zf}{ZF}{Zero Forcing}
\newacronym{lmmse}{LMMSE}{Linear Minimum Mean Square Error}
\newacronym{rmse}{RMSE}{Root Mean Square Error}
\newacronym{mc}{MC}{Monte Carlo}
\newacronym{ser}{SER}{Symbol Error Rate}
\newacronym{mae}{MAE}{Mean Absolute Error}
\newacronym{af}{AF}{Ambiguity Function}
\newacronym{elmmse}{ELMMSE}{Ergodic Linear Minimum Mean Square Error}
\newacronym{nda}{NDA}{Non-Data-Aided}
\newacronym{da}{DA}{Data-Aided}
\newacronym{zzb}{ZZB}{Ziv-Zakai Bound}
\definecolor{myBlack}{HTML}{565656}
\definecolor{myGreen}{HTML}{5bb900}
\definecolor{myGreen2}{HTML}{005e00}
\definecolor{myBlue}{HTML}{005ae9}
\definecolor{myBlue2}{HTML}{009be8}
\definecolor{myRed}{HTML}{d7001d}
\definecolor{myOrange}{HTML}{ffa409}
\definecolor{myPurple}{HTML}{6000e2}
\definecolor{myPink}{HTML}{F03080}

\definecolor{myGreenArrow1}{HTML}{70b500}
\definecolor{myGreenArrow2}{HTML}{0d8800}

\definecolor{Pcolor}{HTML}{828282}
\definecolor{PDcolor}{HTML}{2d2d2d}
\definecolor{MMLacolor}{HTML}{005acb}
\definecolor{MMLfastcolor}{HTML}{0086ff}
\definecolor{DDcentrcolor}{HTML}{ffa409}
\definecolor{DDdistrcolor}{HTML}{d7001d}
\definecolor{SDDcentrcolor}{HTML}{d454ff}
\definecolor{SDDdistrcolor}{HTML}{6000e2}

\theoremstyle{plain}
\newtheorem{assumption}{Assumption}

\newtheorem{proposition}{Proposition}[section]

\newtheorem{remark}{Remark}[section]

\def\BibTeX{{\rm B\kern-.05em{\sc i\kern-.025em b}\kern-.08em
    T\kern-.1667em\lower.7ex\hbox{E}\kern-.125emX}}

\begin{document}

\title{Localization in OFDM Passive Distributed Antenna Systems \\ with Pilots and Unknown Data Payloads: \\ A Marginal Maximum Likelihood Approach
\thanks{Mathieu Reniers is a Research Fellow of the Fonds de la Recherche
Scientifique - FNRS.
}
}

\author{
    Mathieu Reniers\orcidlink{0009-0007-3811-6112},
    Martin Willame\orcidlink{0000-0002-7107-6198}, 
    Jérôme Louveaux\orcidlink{0000-0003-2557-1857}, 
    Luc Vandendorpe\orcidlink{0000-0003-4958-8848}, \\
    ICTEAM, UCLouvain - Louvain-La-Neuve, Belgium.
    \footnotesize{Emails: \{firstname.lastname\}@uclouvain.be}
}

\bstctlcite{MyBSTcontrol}

\maketitle

\begin{abstract}
    Integrated Sensing and Communications (ISAC) is emerging as a key paradigm for future Sixth-Generation (6G) networks, with communication-centric designs favored for their compatibility with existing standards.
Communication signals contain both known deterministic pilot symbols and unknown random data payloads. 
Most localization approaches rely solely on pilots, discarding the position information contained in the data symbols, which constitute the majority of each transmitted frame.
Alternatively, Decision-Directed (DD) approaches exploit data decisions, thereby inherently limiting positioning performance to that of the communication system.
In this paper, we derive a Marginal Maximum Likelihood (MML) estimator that jointly leverages pilot and data payloads without requiring data decoding, enabling operation with high-order constellations and under challenging noise conditions.
We consider an opportunistic scenario in which an Orthogonal Frequency-Division Multiplexing (OFDM) signal transmitted by a User Equipment (UE) is captured by a distributed receiver array. 
Through numerical simulations, we demonstrate that the proposed method achieves superior localization performance compared to existing approaches and consistently converges to the genie bound (where data symbols are assumed perfectly known) at a lower Signal-to-Noise Ratio (SNR) than any DD method.
Furthermore, the proposed method remains robust to constellation size, unlike DD approaches, whose performance degrades with increasing modulation order.
Finally, we provide a computational complexity analysis of the proposed method and the considered baselines, highlighting the impact of system parameters on their respective computational costs.
\end{abstract}


\section{Introduction}

\glsunset{isac}
\IEEEPARstart{I}NTEGRATED Sensing and Communications (ISAC) is likely to become a cornerstone of future \gls{6g} networks and next-generation Wi-Fi standards, enabling mobile systems to sense their environment \cite{meng_cooperative_2024,zheng_radar_2019,mazahir_survey_2021}. 
Integrating communication and sensing into a single platform allows hardware sharing, improving spectrum efficiency \cite{liu_joint_2020} and reducing device cost, size, and power consumption, while potentially enhancing the performance of both services \cite{feng_joint_2020,zhang_overview_2021}.
The \textit{communication-centric} approach—where communication remains the primary function and sensing or radar capabilities are performed by reusing the communication waveform—has emerged as the most promising, owing to its compatibility with existing wireless standards \cite{lu_sensing_2025}.
In this paradigm, the transmitted communication signal comprises pilot symbols, which are \textbf{deterministic} and \textbf{known} to both \gls{tx} and \gls{rx}, alongside data symbols, which carry information and are \textbf{random} and \textbf{unknown} at the \gls{rx}.

\subsection{Related Works}

\subsubsection{Pilot-Only Localization}
The traditional approach for \gls{isac} radar sensing relies solely on pilot or reference signals dedicated to this purpose \cite{wei_5g_2023,golzadeh_joint_2024,wei_multiple_2024}. 
Wei et al.\ \cite{wei_5g_2023} investigate radar sensing performance using \gls{5g-nr} \gls{prs}.
This work was extended to irregular \gls{prs} resource patterns for high-mobility scenarios in \cite{golzadeh_joint_2024} and to collaborative multi-reference-signal schemes in \cite{wei_multiple_2024}.
However, these approaches completely overlook the localization information contained in the data payloads, which constitute the majority of each transmitted frame (typically $\SI{75}{\percent}$--$\SI{97}{\percent}$ in \gls{5g-nr} \cite{3gpp_5g_2025}).

\subsubsection{Characterizing Sensing under Random Data Signals}
To address this limitation, several studies instead characterize the sensing performance under random data signals rather than deterministic pilots \cite{liu_cp-ofdm_2025,keskin_fundamental_2025,du_reshaping_2024}, typically assuming a monostatic configuration where the \gls{srx} has perfect knowledge of the transmitted data symbols. 
Liu et al.\ \cite{liu_cp-ofdm_2025} show that, among all communication waveforms with \glsdesc{cp}, \gls{ofdm} is optimal for minimizing ranging sidelobes with \gls{qam}, strengthening its candidacy for \gls{6g}—and motivating its adoption in this work.
Keskin et al.\ \cite{keskin_fundamental_2025} reveal a fundamental \textit{time-frequency trade-off} in monostatic \gls{ofdm} \gls{isac}: high-order \gls{qam} improves communication rates but degrades sensing due to elevated sidelobe levels from the non-constant modulus data \cite{du_reshaping_2024}, whereas low-order \gls{psk} minimizes sidelobes at the cost of reduced spectral efficiency. 
This trade-off has motivated \glsdesc{pcs} approaches aimed at jointly optimizing both functions \cite{du_reshaping_2024}.
While these studies provide valuable insights into the impact of data randomness on sensing, perfect knowledge of the transmitted data at the \gls{srx} is not attainable in multistatic nor passive sensing configurations.

\subsubsection{Decision-Directed Philosophy}
A natural way to exploit the full transmitted frame is to demodulate the data payloads and treat the decisions as additional pilots, following a \gls{dd} philosophy. 
Wypich and Zielinski \cite{wypich_ofdm-based_2025} demonstrate that in \gls{ofdm}-based passive radar, reconstructed data effectively supplements pilots under low \gls{ser} conditions, while performance reduces to the pilot-only case under challenging channel conditions or high-order modulations—a finding validated experimentally in \cite{wypich_experimental_2026}. 
Brunner et al. \cite{brunner_bistatic_2025} propose a complete bistatic \gls{isac} framework accounting for synchronization and characterize sensing performance over the full frame, including the impact of decoding failures and residual synchronization errors.
To enhance sensing performance at a modest spectral efficiency cost, Henninger et al. \cite{henninger_hybrid_2026} propose a resource allocation scheme placing lower-order constellation symbols as pseudo-pilots.
The \gls{dd} methodology has also been integrated into iterative frameworks exploiting the mutual benefits of \textit{communication-aided sensing} and \textit{sensing-aided communication} \cite{zhao_joint_2024,keskin_bridging_2025-1}. 
This approach was first developed for single-carrier bistatic systems by Zhao et al. in \cite{zhao_joint_2024} and subsequently extended to \gls{ofdm} by Keskin et al. in \cite{keskin_bridging_2025-1}.
While these methods achieve low localization errors under favorable \gls{snr} conditions, they incur high computational complexity due to forward error correction decoding or iterative processing, and their localization performance is inherently limited by that of the communication system.

\subsubsection{Maximum Likelihood Philosophy}
An alternative to \gls{dd} is to treat the unknown data symbols as \textbf{\glspl{np}}, i.e., parameters that affect the distribution of the observations used for estimation but are not of direct interest, as opposed to the \textbf{\glspl{poi}}.
A natural way to eliminate the dependence on the \glspl{np} is to marginalize the likelihood over their prior distribution, i.e., to integrate the conditional likelihood of the observations with respect to the \gls{np} prior.
This approach is referred to as \gls{mml}, or sometimes as \gls{iml} in the literature \cite{berger_integrated_1999}.
In the case of data symbols, this amounts to evaluating the conditional likelihood over all constellation points, rather than making hard symbol decisions as in \gls{dd}.
This formulation falls within the class of \gls{nda} estimation methods, where unknown transmitted symbols are treated as random variables and handled statistically rather than pre-estimated\footnotemark.
\footnotetext{
There is a subtle terminology shift in the recent literature.
Originally, \gls{nda} refers to estimation using unknown symbols (i.e., data), while \gls{da} refers to the use of known symbols (i.e., pilots).
However, some recent works use \gls{da} terminology \cite{mensing_data-aided_2009,gupta_data-aided_2025,keskin_bridging_2025-1} to describe methods following a \gls{dd} philosophy.
}
While \gls{nda} estimators have been widely studied for communication purposes, their application to positioning remains sparse \cite{monfared_iterative_2020,graff_ofdm-based_2026}.
Monfared et al. \cite{monfared_iterative_2020} propose an iterative \gls{mml}-based \gls{nda} algorithm using \glsdesc{aoa} measurements for sensor localization with a \glsdesc{gfsk} waveform.
Graff and Humphreys \cite{graff_ofdm-based_2026} derive the \gls{mml} estimator for \gls{ofdm}-based positioning, compare pilot-only, \gls{dd}, and \gls{mml} range estimation performance against a derived \glsdesc{zzb}, and apply the framework to simulated low-earth orbit satellite channels.
However, their work is limited to a single-antenna receiver where only the data symbols are treated as \glspl{np}, while the timing offset, phase offset, and channel gain are either explicitly estimated or assumed known.
In contrast, the present work extends this framework to a multi-antenna distributed array operating without phase synchronization, treating the channel coefficient at each antenna as an additional \gls{np}, thereby broadening the scope of the \gls{mml} approach.
Furthermore, while their work reports a computational complexity of $\Compl(\#\constSet)$, where $\#\constSet$ denotes the constellation size, we introduce an analytical acceleration that reduces it to $\Compl(\sqrt{\#\constSet})$.

Another approach to handle \glspl{np} is \gls{jml} estimation, where both the \glspl{poi} and \glspl{np} are \textit{jointly} estimated.
In our previous work \cite{reniers_joint_2026}, we applied this method to a source localization scenario restricted to a co-located \gls{ula} operating under \gls{ff} conditions, limitations that the present work overcomes.

\subsection{Contributions}
To address the aforementioned limitations, we develop an \gls{mml} framework for an ``uplink''\footnotemark scenario where the transmitted communication signal from a \gls{ue} is captured by an \textbf{opportunistic \gls{srx}} consisting of \textbf{multiple distributed nodes} whose sole objective is to localize the source.
Since the nodes are \textbf{not phase-synchronized}, the random unknown channel coefficients are treated as additional \glspl{np}, alongside the data symbols.
Given the distributed nature of the \gls{srx}, the spherical nature of the wavefront can no longer be ignored.
Our main contributions are summarized as follows:
\begin{itemize}
    \item We develop the \gls{mml} framework for uplink localization with a distributed, non-phase-synchronized \gls{srx}, treating both data symbols and channel coefficients at all nodes as \glspl{np}.
    The optimal solution is derived and shown to be computationally intractable, motivating the development of practical approximations.
    \item We propose a computationally feasible approximation of the \gls{mml} solution.
    Through \gls{mc} simulations, we demonstrate superior localization performance over existing methods, with two key findings: the \textbf{proposed method is robust to constellation size}—unlike \gls{dd} methods, whose performance degrades with increasing modulation order—and \textbf{converges to the $\PDmethod$ bound} (i.e., data assumed perfectly known at the \gls{srx}) \textbf{at lower \gls{snr} than competing \gls{dd} approaches}.
    \item Additionally, we derive a closed-form computational acceleration of the proposed method for $\constSize$-\gls{qam} constellations, reducing the complexity from $\Compl(\constSize)$ to $\Compl(\sqrt{\constSize})$ while yielding the exact same solution.
    This acceleration is applicable to other likelihood-based problems and, to the best of the authors' knowledge, has not been previously reported in the literature.
    \item We provide a comprehensive analysis of the proposed method against the considered baselines, examining the impact of system parameters on both localization performance and computational requirements, supported by theoretical insights and simulation results.
\end{itemize}

\footnotetext{This terminology is used by analogy with the conventional communication uplink, and is used here interchangeably with source localization.}
\subsection{Structure of the Paper}
The remainder of the paper is organized as follows.
Section \ref{sec:SystemModel} describes the system model, including the scenario, channel model, and signal definitions.
In Section \ref{sec:MMLEstimator}, we derive the optimal \gls{mml} estimator, establish its computational intractability, introduce a feasible approximation, and present the analytical acceleration for \gls{qam} constellations.
Section \ref{sec:NumericalResults} presents the \gls{mc} simulation results and investigates the impact of system parameters on localization performance, as well as the impact of data demodulation errors on \gls{dd}-based localization.
Section \ref{sec:Complexity} provides a detailed complexity analysis.
Finally, Section \ref{sec:Conclusion} concludes the paper and outlines directions for future work.\\
Throughout this paper, some key observations and discussions are presented in \textbf{Remark} environments for ease of reference.

\subsection{Notations}
Scalars, vectors, matrices and tensors are respectively denoted by $a$, $\boldsymbol{a}$, $\boldsymbol{A}$ and $\boldsymbol{\mathcal{A}}$, and when expressed as function of parameters, written as $a(\cdot)$, $\boldsymbol{a}(\cdot)$, $\boldsymbol{A}(\cdot)$ and  $\boldsymbol{\mathcal{A}}(\cdot)$.
The $i$-th element of $\boldsymbol{a}$, the $(i,j)$-th element of $\boldsymbol{A}$ and the $(i,j,k)$-th element of $\boldsymbol{\mathcal{A}}$ are expressed as $\boldsymbol{a}[i]$, $\boldsymbol{A}[i,j]$ and $\boldsymbol{\mathcal{A}}[i,j,k]$. 
The absolute value, the vector $\ell_2$ norm and the Frobenius norm are denoted by $\abs{a}$, $\norm{\boldsymbol{a}}$ and $\frob{\boldsymbol{A}}$, respectively.
The $K \times K$ identity matrix is written as $\I{K}$.
The real and complex sets are denoted by $\R$ and $\C$, the complex conjugate is defined by $a^*$, and $\jc \triangleq \sqrt{-1}$. 
The cardinality of a finite set $\mathcal{S}$ is given by $\#\mathcal{S}$.
The likelihood function is represented by~$\Likelihood(\cdot)$.
Given a true quantity $\theta$, candidate and final estimates are indicated by $\widetilde{\theta}$ and $\widehat{\theta}$, respectively.
Finally, the expectation over $\theta$ is notated $\E_{\theta} \{ \cdot \}$.

\section{System Model}\label{sec:SystemModel}
\subsection{Scenario}
This paper investigates the problem of localizing a single-antenna \gls{ue} at position $\UEpos \in \R^{2}$, transmitting an \gls{ofdm} uplink signal received by an \gls{srx} consisting of $\numRX$ single-antenna \gls{rx} nodes at positions $\{\RXpos{\RXindex}\}_{\RXindex=0}^{\numRX-1} \in \R^{\numRX \times 2}$ and forming a \gls{das}.
Note that the proposed framework can be straightforwardly extended to $\R^{3}$ without loss of generality.
An illustration of the scenario is presented in \autoref{fig:scenario}.
Note that the node positions $\{\RXpos{\RXindex}\}_{\RXindex=0}^{\numRX-1}$ are known to the \gls{cpu}, which estimates $\UEpos$.

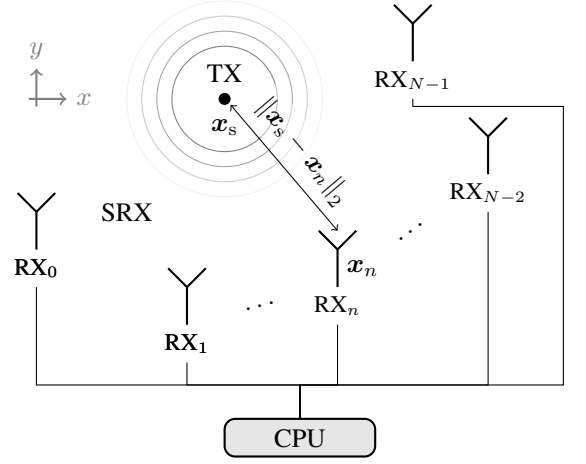
\begin{figure}[t]
    \centering
        \resizebox{1.0\linewidth}{!}{%
        \def\SRXCOLOR{black}
\def\TXCOLOR{black}

\newcommand{\antenna}[2][1]{
    \draw[thick, \SRXCOLOR] #2 -- ++(0,0.5*#1);
    \draw[thick, \SRXCOLOR] #2 ++(0,0.5*#1) -- ++(-0.25*#1,0.25*#1);
    \draw[thick, \SRXCOLOR] #2 ++(0,0.5*#1) -- ++(0.25*#1,0.25*#1);
}
\def\ueX{-0.5}  
\def\ueY{0} 
\def\ueradius{0.08} 
\def\antennaheight{0.75}    
\def\antennapositions{(-3,-2), (-1,-3), (1,-2.5), (3,-1), (2,0.5)}
\def\cpuY{-4.5}  
\def\cpuX{0.5}    

\begin{tikzpicture}
    \coordinate (ue) at (\ueX,\ueY);
    \fill[\TXCOLOR] (ue) circle (\ueradius);
    \node[above, \TXCOLOR] at ($(ue)+(0,0.1)$) {TX};
    \node[below, \TXCOLOR] at ($(ue)+(0,-0.15)$) {$\UEpos$};

    \draw[\TXCOLOR!50!white] (ue) circle (0.7);
    \draw[\TXCOLOR!35!white] (ue) circle (0.9);
    \draw[\TXCOLOR!20!white] (ue) circle (1.1);
    \draw[\TXCOLOR!5!white] (ue) circle (1.3);
    \foreach [count=\i from 0] \pos in \antennapositions {
        \antenna[1.0]{\pos}
        \path \pos coordinate (ant\i);  
    }
    \foreach [count=\i from 0] \pos in \antennapositions {
            \ifnum\i<2
        \node[below, \SRXCOLOR] at \pos {\small RX$_{\i}$};
        \node[below, \SRXCOLOR] at \pos (rx\i) {\small RX$_{\i}$};
    \fi
    }
    \path let \p1=(ant1), \p2=(ant2), \n1={atan2(\y2-\y1,\x2-\x1)} in
    node[rotate=\n1] at ($(ant1)!0.5!(ant2)$) {$\cdots$};
    \path let \p1=(ant2), \p2=(ant3), \n1={atan2(\y2-\y1,\x2-\x1)} in
    node[rotate=\n1] at ($(ant2)!0.5!(ant3)$) {$\cdots$};
    \node[below, \SRXCOLOR] at (ant2) (rx2) {\small RX$_{n}$};
    \node[right, \SRXCOLOR] at ($(ant2)+(0,+0.25)$) {$\RXpos{\RXindex}$};
    \node[below, \SRXCOLOR] at (ant3) (rx3) {\small RX$_{\numRX - 2}$};
    \node[below, \SRXCOLOR] at (ant4) (rx4) {\small RX$_{\numRX - 1}$};
    
    \node[draw=\SRXCOLOR, thick, rectangle, rounded corners, fill=\SRXCOLOR!10, text=\SRXCOLOR, minimum width=2cm, minimum height=0.5cm] at (\cpuX,\cpuY) (cpu) {CPU};
    \draw[\SRXCOLOR,thin] ($(ant0)+(0,-0.5)$) -- ++(0,-1.3) -| (cpu.north);
    \draw[\SRXCOLOR,thin]  ($(ant1)+(0,-0.5)$) -- ++(0,-0.3) -| (cpu.north);
    \draw[\SRXCOLOR,thin]  ($(ant2)+(0,-0.5)$) -- ++(0,-0.8) -| (cpu.north);
    \draw[\SRXCOLOR,thin]  ($(ant3)+(0,-0.5)$) -- ++(0,-2.3) -| (cpu.north);
    \draw[\SRXCOLOR,thin]  ($(ant4)+(0,-0.5)$) -- ++(0.0,-0.1) -- ++(2.0,0.0) -- ++(0,-3.7) -| (cpu.north);

    \draw[<->,sloped] ($(ant2)+(0,\antennaheight)$) -- ($(ue)+(\ueradius,-\ueradius)$) 
            node[pos=0.5, align=center,above=0.02cm, fill=white, fill opacity=0.5, text opacity=1] {$\norm{\UEpos - \RXpos{n}}$};
    
    \draw[thick, ->, gray] (-3.1,0) -- ++(0.5,0) node[right]{$x$};
    \draw[thick, ->, gray] (-3,-0.1) -- ++(0,0.5) node[above]{$y$};

    \node[\SRXCOLOR] at (-1.8,-1.5) {SRX};

\end{tikzpicture}
    }
    \caption{Illustration of the considered scenario. 
    The \gls{ue} (i.e., the \gls{tx}) transmits an uplink signal to a \gls{das} (i.e., the \gls{srx}) consisting of $\numRX$ single-antenna \gls{rx} nodes. 
    A \gls{cpu} collects the received signals and performs localization.
    }
    \label{fig:scenario}
\end{figure}

The \gls{ue} transmits $\numP$ pilot and $\numD$ data \gls{ofdm} symbols across $\numF$ subcarriers, with a subcarrier spacing of $\Fspacing$, yielding a total bandwidth of $\BW \triangleq \numF \Fspacing$.
The frequency corresponding to the first subcarrier is denoted by $\carrierF$, with associated wavelength $\carrierWl$ and wavenumber $\carrierWn \triangleq 2\pi \carrierF / c = 2\pi / \carrierWl$, where $c$ represents the speed of light.
Pilot symbols are given by $\pilotMat \in \C^{\numF \times \numP}$ and may consist of arbitrary complex sequences.
Data symbols are given by $\dataMat \in \constSet^{\numF \times \numD} \subset \C^{\numF \times \numD}$, where each element $\dataMat[\Findex,\Dindex]$ belongs to a given constellation $\constSet_{\constMap(\Findex,\Dindex)} \in \constSet$, selected from the set $\constSet$.
The constellation assignment across frquency and time dimensions is denoted $\constMap(\Findex,\Dindex)$ and is determined by the \gls{tx}.
It is typically defined in the preamble containing transmission metadata.
To simplify notation, we adopt the following shorthand: \begingroup \small$\dataMat \in \constSet_{\constMap}^{\numF \times \numD} \triangleq \{\dataMat[\Findex,\Dindex] \in \constSet_{\constMap(\Findex,\Dindex)}\}_{\Findex,\Dindex=0,0}^{\numF-1,\numD-1}$\endgroup.
The symbol transmission follows the resource grid illustrated in \autoref{fig:resource_grid}.
Note that other resource allocations can be accommodated; specifically, the time-domain distribution of pilot and data symbols has no impact, as the channel is assumed constant across the time dimension (see Section~\ref{sec:SystemModelChannel}), whereas a uniform allocation across the frequency axis is required.

\begin{figure}[t]
    \centering
    \resizebox{1.0\linewidth}{!}{%
        \def\numPilots{4}
\def\numData{14}
\pgfmathsetmacro{\numTimeSlots}{\numPilots+\numData}   
\def\numSubcarriers{9}   
\def\cellSize{0.25}       

\begin{tikzpicture}

    \pgfmathsetmacro{\numPilotsMinusOne}{\numPilots-1}
    \pgfmathsetmacro{\numTimeSlotsMinusOne}{\numTimeSlots-1}
    \pgfmathsetmacro{\numSubcarriersMinusOne}{\numSubcarriers-1}
    
    \foreach \x in {0,...,\numPilotsMinusOne} {
        \foreach \y in {0,...,\numSubcarriersMinusOne} {
            \fill[black!10] (\x*\cellSize, \y*\cellSize) rectangle ({(\x+1)*\cellSize}, {(\y+1)*\cellSize});
        }
    }
    
    \foreach \x in {\numPilots,...,\numTimeSlotsMinusOne} {
        \foreach \y in {0,...,\numSubcarriersMinusOne} {
            \fill[black!20] (\x*\cellSize, \y*\cellSize) rectangle ({(\x+1)*\cellSize}, {(\y+1)*\cellSize});
        }
    }

    \foreach \x in {0,...,\numTimeSlots} {
        \draw[black!30] (\x*\cellSize, 0) -- (\x*\cellSize, \numSubcarriers*\cellSize);
    }
    \foreach \y in {0,...,\numSubcarriers} {
        \draw[black!30] (0, \y*\cellSize) -- (\numTimeSlots*\cellSize, \y*\cellSize);
    }

    \draw[->, thick] (-0.75*\cellSize, 0) -- (\numTimeSlots*\cellSize + \cellSize, 0) node[pos=0.4, below=0.5cm] {Time};
    \draw[->, thick] (0, -0.75*\cellSize) -- (0, \numSubcarriers*\cellSize + \cellSize) node[pos=0.5, left=1.5cm, rotate=90, anchor=center] {Frequency};

    \draw[<->, thick] (0, -\cellSize) -- (\numPilots*\cellSize, -\cellSize) node[midway, below] {$\numP\Tspacing$};
    \draw[<->, thick] (\numPilots*\cellSize, -\cellSize) -- (\numTimeSlots*\cellSize, -\cellSize) node[midway, below] {$\numD\Tspacing$};
    \draw[<->, thick] (-\cellSize,0) -- (-\cellSize, \numSubcarriers*\cellSize) node[midway, left] {$\numF\Fspacing$}
    node[pos=0, left=0.1cm] {$\carrierF$};

    \pgfmathsetmacro{\pilotCenter}{\numPilots*\cellSize/2}
    \pgfmathsetmacro{\dataCenter}{(\numPilots + \numTimeSlots)*\cellSize/2}
    \pgfmathsetmacro{\verticalCenter}{\numSubcarriers*\cellSize/2}
    
    \node[] at (\pilotCenter, \verticalCenter) {$\pilotMat$};
    \node[] at (\dataCenter, \verticalCenter) {$\dataMat$};

    \draw[<->, thick]
    (\numTimeSlots*\cellSize - 2*\cellSize, \numSubcarriers*\cellSize + 0.75*\cellSize)
    --
    (\numTimeSlots*\cellSize - 1*\cellSize, \numSubcarriers*\cellSize + 0.75*\cellSize)
    node[midway, above] {$\Tspacing$};

    \draw[<->, thick]
    (\numTimeSlots*\cellSize + 0.75*\cellSize, \numSubcarriers*\cellSize - 2*\cellSize)
    --
    (\numTimeSlots*\cellSize + 0.75*\cellSize, \numSubcarriers*\cellSize - 1*\cellSize)
    node[midway, right] {$\Fspacing$};


    \def\DataIndex{1}
    \def\SubcarrierIndex{3}
    \pgfmathsetmacro{\dotX}{(\numPilots + \DataIndex + 0.5)*\cellSize}
    \pgfmathsetmacro{\dotY}{(\SubcarrierIndex + 0.5)*\cellSize}

    \fill[myBlue] (\dotX, \dotY) circle (1.5pt);

    \pgfmathsetmacro{\gridTop}{\numSubcarriers*\cellSize}
    \pgfmathsetmacro{\constellationBottom}{\gridTop + 0.3}
    \pgfmathsetmacro{\constellationCenterX}{\dotX - 2*\cellSize}
    \pgfmathsetmacro{\constellationCY}{\constellationBottom + 0.45}

    \draw[myBlue, thin, dashed] (\dotX, \dotY) node[below, font=\tiny] {$(\Findex,\Dindex)$} -- (\constellationCenterX, \constellationBottom) ;

    \def\qamScale{0.1}
    \def\qamDotSize{1.2pt}
    \pgfmathsetmacro{\qamHalf}{3*\qamScale}

    \node[above, font=\footnotesize] at (\constellationCenterX, \constellationBottom + 0.93)
        {$\constSet_{\constMap(\Findex,\Dindex)}$};

    \draw[->, black!50, thin]
        ({\constellationCenterX - \qamHalf - 0.2}, {\constellationCY})
        --
        ({\constellationCenterX + \qamHalf + 0.2}, {\constellationCY});
    \draw[->, black!50, thin]
        ({\constellationCenterX}, {\constellationBottom + 0.02})
        --
        ({\constellationCenterX}, {\constellationBottom + 0.93});

    \pgfmathsetmacro{\selectedQI}{3}    
    \pgfmathsetmacro{\selectedQQ}{3}  

    \pgfmathsetmacro{\selectedI}{2*\selectedQI - 3}
    \pgfmathsetmacro{\selectedQ}{2*\selectedQQ - 3}

    \foreach \xi in {-3,-1,1,3} {
        \foreach \yi in {-3,-1,1,3} {
            \pgfmathsetmacro{\px}{\constellationCenterX + \xi*\qamScale}
            \pgfmathsetmacro{\py}{\constellationCY + \yi*\qamScale}
            \pgfmathsetmacro{\isSelected}{(\xi == \selectedI) && (\yi == \selectedQ) ? 1 : 0}
            \ifnum\isSelected=1
                \fill[myBlue] (\px, \py) circle (2pt);
                \draw[myBlue, thin] (\px, \py) circle (3pt) node[above right=1pt, font=\footnotesize] {$\dataMat[\Findex,\Dindex]$};
            \else
                \fill[black] (\px, \py) circle (\qamDotSize);
            \fi
        }
    }


    \def\DataIndex{9}
    \def\SubcarrierIndex{6}
    \pgfmathsetmacro{\dotX}{(\numPilots + \DataIndex + 0.5)*\cellSize}
    \pgfmathsetmacro{\dotY}{(\SubcarrierIndex + 0.5)*\cellSize}

    \fill[myBlue2] (\dotX, \dotY) circle (1.5pt);

    \pgfmathsetmacro{\gridTop}{\numSubcarriers*\cellSize}
    \pgfmathsetmacro{\constellationBottom}{\gridTop + 0.5}
    \pgfmathsetmacro{\constellationCenterX}{\dotX + 1*\cellSize}
    \pgfmathsetmacro{\constellationCY}{\constellationBottom + 0.45}

    \draw[myBlue2, thin, dashed] (\dotX, \dotY) node[below, font=\tiny] {$(\Findex',\Dindex')$} -- (\constellationCenterX, \constellationBottom);

    \def\qamScale{0.2}   
    \def\qamDotSize{1.2pt}
    \pgfmathsetmacro{\qamHalf}{1*\qamScale}

    \node[above, font=\footnotesize] at (\constellationCenterX, \constellationBottom + 0.93)
        {$\constSet_{\constMap(\Findex',\Dindex')}$};

    \draw[->, black!50, thin]
        ({\constellationCenterX - \qamHalf - 0.15}, {\constellationCY})
        --
        ({\constellationCenterX + \qamHalf + 0.15}, {\constellationCY});
    \draw[->, black!50, thin]
        ({\constellationCenterX}, {\constellationBottom + 0.02})
        --
        ({\constellationCenterX}, {\constellationBottom + 0.93});

    \pgfmathsetmacro{\selectedQI}{0}   
    \pgfmathsetmacro{\selectedQQ}{0}   

    \pgfmathsetmacro{\selectedI}{2*\selectedQI - 1}
    \pgfmathsetmacro{\selectedQ}{2*\selectedQQ - 1}

    \foreach \xi in {-1,1} {
        \foreach \yi in {-1,1} {
            \pgfmathsetmacro{\px}{\constellationCenterX + \xi*\qamScale}
            \pgfmathsetmacro{\py}{\constellationCY + \yi*\qamScale}
            \pgfmathsetmacro{\isSelected}{(\xi == \selectedI) && (\yi == \selectedQ) ? 1 : 0}
            \ifnum\isSelected=1
                \fill[myBlue2] (\px, \py) circle (2pt);
                \draw[myBlue2, thin] (\px, \py) circle (3pt) node[below left=1pt, font=\footnotesize] {$\dataMat[\Findex',\Dindex']$};
            \else
                \fill[black] (\px, \py) circle (\qamDotSize);
            \fi
        }
    }

\end{tikzpicture}
    }
    \caption{Illustration of the considered resource grid. 
    }
    \label{fig:resource_grid}
\end{figure}
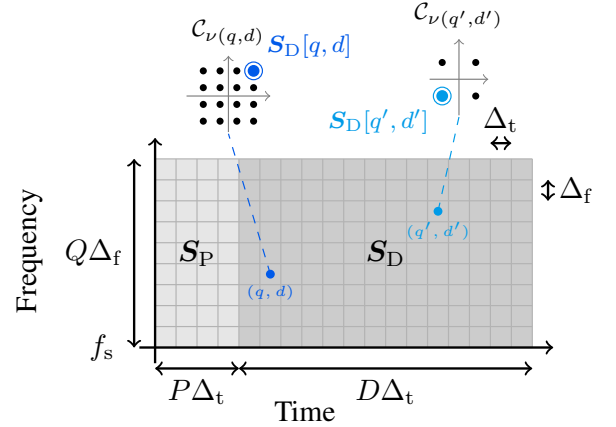

The following assumption is adopted for the considered scenario:
\begin{assumption}\label{ass:stillness}
    The \gls{ue} and the background are assumed to be stationary during the frame duration $\T \triangleq (\numP+\numD)\Tspacing$, with $\Tspacing$ being the \gls{ofdm} symbol duration. Hence, no Doppler effect is considered.
\end{assumption}
For the typical \gls{6g} \gls{isac} parameters used in the simulations (see \autoref{tab:parameters}), this yields\footnotemark\ $\Tspacing \approx \Fspacing^{-1} = \SI{16.67}{\micro\second}$ and $\T=(\numP+\numD)\Tspacing \approx \SI{0.6}{\milli\second}$, which is sufficiently short to justify the stationarity assumption over a single localization update.
\footnotetext{Neglecting the cyclic prefix (of length $\numCP$), from $\Tspacing = \frac{\numF + \numCP}{\numF \Fspacing} \approx \Fspacing^{-1}$ for $\numF \gg \numCP$.}

\subsection{Channel model and observations}\label{sec:SystemModelChannel}
The channel model is specified under the following assumptions:
\begin{assumption}\label{ass:los}
    Only \gls{los} propagation is considered.
\end{assumption}
\begin{assumption}\label{ass:synchro}
    Time synchronization is already achieved at each \gls{rx} node, and local oscillators are matched (i.e., no carrier frequency offset) between the \gls{ue} and all \gls{rx} nodes.
\end{assumption}

The channel matrix is defined as $\channelMat(\UEpos) \in \C^{\numRX \times \numF}$.
Following the aforementioned assumptions, the channel coefficient at the $\RXindex$-th \gls{rx} node and on the $\Findex$-th subcarrier can be written as (see, for instance, \cite{sakhnini_near-field_2022})
\begin{align}
    \channelMat(\UEpos)[\RXindex,\Findex] & = \channel[\RXindex] e^{-\jc 2 \pi (\carrierF + \Findex \Fspacing) 
    \norm{\UEpos - \RXpos{\RXindex}}/c} \\
    & = \underbrace{\channel[n] e^{-\jc \carrierWn \norm{\UEpos - \RXpos{\RXindex}}}}_{\triangleq \channelcoeff[n]} 
    \underbrace{e^{-\jc \carrierWn  \norm{\UEpos - \RXpos{\RXindex}} \Findex \frac{\Fspacing}{\carrierF}}}_{\triangleq \stMat(\UEpos)[\RXindex,\Findex]}, \label{eq:channel_model}
\end{align}
where $\channel \in \C^{\numRX \times 1}$ contains the random channel coefficient at all nodes, accounting for propagation losses and random phases at each \gls{rx}; $\channelcoeff \in \C^{\numRX \times 1}$ additionally includes the propagation phases associated with the reference frequency; and $\stMat(\UEpos) \in \C^{\numRX \times \numF}$ denotes the steering matrix (i.e., the part of the model that carries the localization information).
\begin{remark}
We consider a distributed \gls{srx} operating \textbf{without phase synchronization}, such that each node $\RXindex$ exhibits an independent random phase, which is absorbed into $\channel[\RXindex]$. 
Although phase synchronization of distributed arrays has been studied in the literature \cite{nanzer_distributed_2021,rashid_frequency_2023-1}, the opportunistic \gls{srx} considered here is assumed to lack the required infrastructure, rendering phase calibration impractical or prohibitively costly due to the associated iterative procedures.
\end{remark}

\begin{remark}
This model assumes that $\channel$ and therefore $\channelcoeff$ are constants across the time and frequency domains, which requires $\T < \T_{\mathrm{c}}$ and $\BW < \BW_{\mathrm{c}}$, where $\T_{\mathrm{c}}$ and $\BW_{\mathrm{c}}$ are the coherence time and bandwidth.
These conditions are directly satisfied by Assumptions \ref{ass:stillness} and~\ref{ass:los}.
\end{remark}

At the \gls{srx}, the observations associated with pilot and data components are respectively denoted as $\pilotObs \in \C^{\numRX \times \numF \times \numP}$ and $\dataObs \in \C^{\numRX \times \numF \times \numD}$ and are expressed as
\begin{align}
    \pilotObs [\RXindex, \Findex, \Pindex ] &= \channelMat(\UEpos)[\RXindex,\Findex] \pilotMat[\Findex,\Pindex] + \pilotAWGN[\RXindex, \Findex, \Pindex], \label{eq:pilotObs}\\
    \dataObs[\RXindex, \Findex, \Dindex] &= \channelMat(\UEpos)[\RXindex,\Findex] \dataMat[\Findex,\Dindex] + \dataAWGN[\RXindex, \Findex, \Dindex], \label{eq:dataObs}
\end{align}
where $\pilotAWGN \in \C^{\numRX \times \numF \times \numP}$ and $\dataAWGN \in \C^{\numRX \times \numF \times \numD}$ are the \gls{awgn} terms that are independent across all dimensions, with each element following a complex normal distribution $\CN(0,\noiseVar)$.

\section{Marginal Maximum Likelihood Estimator}\label{sec:MMLEstimator}

The objective is to estimate the \gls{poi} (i.e., the \gls{ue} position $\UEpos$), from both the pilot and data observations. 
Since data symbols $\dataMat$ are unknown at the \gls{srx}, they act as \gls{np}, along with channel coefficients $\channelcoeff$.
To extract position information from $\pilotObs$ and $\dataObs$, one must eliminate the dependence on these \glspl{np}.
In this paper, we investigate an \gls{mml} approach, where \glspl{np} uncertainty is eliminated by marginalizing the conditional likelihood over their prior distributions.\\

The \gls{mml} estimation problem is expressed as
\begin{align}
     \UEposEst^{\MMLo} &=  \argmax_{\UEposTest} \E_{\dataMat,\channelcoeff} 
    \left\{ \Likelihood(\pilotObs, \dataObs | \dataMat, \channelcoeff; \UEposTest) \right\} \\
    & = \argmax_{\UEposTest}
     \underbrace{\ExpData{
        \underbrace{\ExpChannel{
    \Likelihood(\pilotObs, \dataObs | \dataMat, \channelcoeff; \UEposTest)}}_{= \Likelihood(\pilotObs, \dataObs | \dataMat; \UEposTest)}
    }}_{= \Likelihood(\pilotObs, \dataObs; \UEposTest)}
    , \label{eq:MMLo_formulation}
\end{align}
where $\Likelihood(\pilotObs, \dataObs | \dataMat, \channelcoeff ; \UEposTest)$ denotes the likelihood or probability of observing $(\pilotObs, \dataObs)$ conditioned on the \textbf{stochastic} \glspl{np} $(\dataMat, \channelcoeff)$ and given the \textit{candidate} \gls{poi} $\UEposTest$.
Equation \eqref{eq:MMLo_formulation} follows from the law of total expectation.
For brevity, we introduce the time-concatenated quantities $\symbolMat \in \C^{\numF \times \numPD}$ and $\Obs \in \C^{\numRX\times \numF \times \numPD }$, where $\numPD \triangleq \numP + \numD$ is the total number of \gls{ofdm} symbols. \\

This section first derives the optimal solution to \eqref{eq:MMLo_formulation} by successively marginalizing over $\channelcoeff$ and $\dataMat$, and shows that it is computationally intractable.
A tractable approximation is then derived by leveraging pilot information to eliminate the channel dependency in the data component.
Finally, an analytical acceleration is developed for \gls{qam} constellations, reducing the computational complexity without any additional approximation.

\subsection{Optimal MML solution}

Marginalizing over $\channelcoeff$ and $\dataMat$ in \eqref{eq:MMLo_formulation} yields the following closed-form expression.
\begin{proposition}\label{prop:MMLo}
    We assume \gls{iid} channel coefficients $\channelcoeff$, each following a circularly symmetric complex Gaussian distribution with variance $\channelVar$, i.e., $\channelcoeff \sim \CN(\boldsymbol{0}, \channelVar \I{\numRX})$, and data symbols drawn independently and uniformly from their respective constellations, i.e., $ \forall \, \Findex,\Dindex \; \dataMat[\Findex, \Dindex] \sim \U_{\constSet_{\constMap(\Findex,\Dindex)}}$.
    Marginalizing out $\channelcoeff$ and $\dataMat$ in \eqref{eq:MMLo_formulation} then yields
    \begin{equation}\label{eq:MMLo_estimator}
        \UEposEst^{\MMLo} = \argmax_{\UEposTest} \sum_{\dataMat \in \constSet_{\constMap}^{\numF \times \numD}} 
        \Likelihood(\Obs | \dataMat; \UEposTest) \Likelihood(\dataMat),
    \end{equation}
    where the conditional likelihood is given by 
    \begin{align}
        & \Likelihood(\Obs | \dataMat; \UEposTest) = C \left( \frac{1}{\noiseVar} \sum_{\Findex,\PDindex} \abs{\symbolMat[\Findex, \PDindex]}^2 + \frac{1}{\channelVar} \right)^{-\numRX} \nonumber \\
        & \quad \times \exp \left( \frac{ \sum_{\RXindex} \abs{
            \frac{1}{\noiseVar} \sum_{\Findex,\PDindex} \Obs^{*}[\RXindex, \Findex,\PDindex] \stMat(\UEposTest)[\RXindex,\Findex]\symbolMat[\Findex,\PDindex]}^2
        }{\frac{1}{\noiseVar} \sum_{\Findex,\PDindex} \abs{\symbolMat[\Findex, \PDindex]}^2 + \frac{1}{\channelVar}} \right), \label{eq:MMLo_channel_integrated}
    \end{align}
   and $C \in \R$ is a constant whose expression is provided in  \hyperref[app:MMLo_proof]{Appendix~\ref*{app:MMLo_proof}}
\end{proposition}
\begin{proof}
    See \hyperref[app:MMLo_proof]{Appendix~\ref*{app:MMLo_proof}}.
\end{proof}

\newpage
\begin{remark}
    \autoref{prop:MMLo} reveals that \underline{spatial coherence} across receivers ($\RXindex$) is lost when marginalizing over the channel coefficients $\channelcoeff$. 
    Indeed, due to the distributed nature of the array, each node $\RXindex$ experiences an independent random complex coefficient $\channel[n]$, so that the received signals are no longer phase-aligned across space. 
    As a result, the contributions of different receivers combine incoherently, which is reflected in \eqref{eq:MMLo_channel_integrated} by the summation over receiver indices outside the squared norm.
    This loss of spatial coherence \textbf{degrades the likelihood function}, as will be discussed in Section~\ref{sec:results_impact_parameters} (see \autoref{fig:AF_Q}).
    In contrast, \underline{frequency} ($\Findex$) and \underline{time} ($\PDindex$) \underline{coherence} are preserved, as the expression remains conditioned on $\dataMat$ and all receivers share the same transmitted symbol sequence.
\end{remark}

Note that the resulting estimator requires knowledge of the noise variance $\noiseVar$, which in practice can be estimated from the \gls{snr} using various existing techniques \cite{pauluzzi_comparison_2000, boumard_novel_2003,huang_investigation_2010}.

\begin{remark}
    Due to the discrete nature of the data symbols and the preservation of time and frequency coherence, the \gls{mml} estimator in \eqref{eq:MMLo_estimator} requires an exhaustive evaluation over all possible symbol combinations $\dataMat \in \constSet_{\constMap}^{\numF \times \numD}$.
This involves $\mathcal{O}(\constSizeMax^{\numF \numD})$ combinations, where $\constSizeMax \triangleq \max_{\Findex,\Dindex} \#\constSet_{\constMap(\Findex,\Dindex)}$ is the largest constellation size.
Even for a modest configuration with $\numD=4$, $\numF=8$, and $16$-\gls{qam}, this amounts to $16^{8\times4} \approx 3.4 \times 10^{38}$ combinations—far beyond any practical computational capability.
This exponential complexity renders the optimal \gls{mml} estimator computationally intractable in practice.
\end{remark}

\subsection{Approximate MML solution}
To enable a tractable solution, we introduce an approximation of the \gls{mml} estimator.
The approximation consists of eliminating the dependency on the channel coefficients $\channelcoeff$ in the data term by substituting a pilot-based estimate constructed for a candidate position $\UEposTest$, denoted $\channelcoeffEstPilots(\UEposTest)$.
The resulting approximate \gls{mml} estimator is given by
\begin{align}
    &\UEposEst^{\MMLa} =  \argmax_{\UEposTest}
      \E_{\dataMat, \channelcoeff}\Big\{ \Likelihood(\pilotObs | \channelcoeff; \UEposTest) \nonumber \\
    & \hspace{3.4cm} \times  \Likelihood(\dataObs | \dataMat; \UEposTest, \channelcoeffEstPilots(\UEposTest)) \Big\} \label{eq:MMLa_firstformulation} \\[0.1em]
    & \textcolor{white}{\UEposEst^{\MMLa}} = \argmax_{\UEposTest}
        \underbrace{\ExpChannel{\Likelihood(\pilotObs | \channelcoeff; \UEposTest)}}_{\text{Pilot term}} \nonumber \\
    & \hspace{2.3cm} \times \underbrace{\ExpData{ \Likelihood(\dataObs | \dataMat; \UEposTest, \channelcoeffEstPilots(\UEposTest))}}_{\text{Data term}},   
    \label{eq:MMLa_formulation}
\end{align}
where \eqref{eq:MMLa_firstformulation} follows from \eqref{eq:MMLo_formulation} by injecting the approximation $\Likelihood(\dataObs | \channelcoeff, \dataMat; \UEposTest) \approx \Likelihood(\dataObs | \dataMat; \UEposTest, \channelcoeffEstPilots(\UEposTest))$.
This approximation eliminates the \gls{np} $\channelcoeff$ dependency through the pilot term, while the $\dataMat$ dependency is handled by marginalization in the data term.
Indeed, the data term now depends on a single \gls{np} term, for which this marginalization can be computed efficiently, as highlighted in Remark~\ref{rem:MMLa_comp}.
The pilot-based channel coefficient estimate is obtained from
\begin{equation}\label{eq:channelcoeffEstPilots_ML}
    \channelcoeffEstPilots(\UEposTest) = \argmax_{\channelcoeffTest \in \C^{\numRX \times 1}} \Likelihood(\pilotObs; \channelcoeffTest, \UEposTest),
\end{equation}
which is separable across receivers $\RXindex$. 
Applying Wirtinger calculus \cite{koor_short_2023} yields
\begingroup\small
\begin{align}
    \channelcoeffEstPilots(\UEposTest)[\RXindex] & = \frac{\sum_{\Findex,\Pindex=0,0}^{\numF-1,\numP-1} \stMat^{*}(\UEposTest)[\RXindex,\Findex] \pilotMat^{*}[\Findex,\Pindex]  \pilotObs[\RXindex, \Findex, \Pindex]}{\sum_{\Findex,\Pindex=0,0}^{\numF-1,\numP-1}\abs{\stMat(\UEposTest)[\RXindex,\Findex] \pilotMat[\Findex,\Pindex]}^2} \\
    & = \frac{1}{\pilotEnergy} \sum_{\Findex=0}^{\numF-1} \stMat^{*}(\UEposTest)[\RXindex,\Findex] \pilotObsEq[\RXindex,\Findex], \label{eq:channelcoeffEstPilots}
\end{align}
\endgroup
where the denominator reduces to the pilot energy $\pilotEnergy \triangleq \frob{\pilotMat}^{2}$ since $\forall \UEposTest \, \abs{\stMat(\UEposTest)[\RXindex,\Findex]}^2 = 1$, and where the ``symbol-equalized'' pilot observations $\pilotObsEq \in \C^{\numRX \times \numF}$ are defined as
\begin{equation}\label{eq:pilotObsEq}
    \pilotObsEq[\RXindex,\Findex] \triangleq \sum_{\Pindex=0}^{\numP-1} \pilotMat^{*}[\Findex,\Pindex]  \pilotObs[\RXindex, \Findex, \Pindex],
\end{equation}
as this quantity will be useful in the following.
For clarity in subsequent expressions, we also define $\channelConstruct(\UEposTest) \in \C^{\numRX \times \numF}$ as the channel estimate constructed for a candidate position $\UEposTest$, combining the pilot-based estimate of the random coefficient 
$\channelcoeff[\RXindex]$ with the steering matrix defined in \eqref{eq:channel_model}:
\begin{equation}
    \channelConstruct(\UEposTest)[\RXindex,\Findex] = \channelcoeffEstPilots(\UEposTest)[\RXindex] \stMat(\UEposTest)[\RXindex,\Findex].
\end{equation}

\begin{proposition}\label{prop:MMLa}
    Assuming the following distributions: $\channelcoeff \sim \CN(\boldsymbol{0}, \channelVar \I{\numRX})$ and $ \forall \Findex,\Dindex \; \dataMat[\Findex, \Dindex] \sim \U_{\constSet_{\constMap(\Findex,\Dindex)}}$, independently, the estimator $\UEposEst^{\MMLa}$ defined in \eqref{eq:MMLa_formulation} admits the closed-form expression \eqref{eq:MMLa_estimator}.
\end{proposition}
\begin{proof}
    See \hyperref[app:MMLa_proof]{Appendix~\ref*{app:MMLa_proof}}.
\end{proof}

\begin{figure*}[t]
    \begingroup \small
    \begin{align}
        \UEposEst^{\MMLa} = \argmax_{\UEposTest} & \left( \frac{\pilotEnergy}{\noiseVar} + \frac{1}{\channelVar} \right)^{-1} 
        \sum_{\RXindex=0}^{\numRX-1} \abs{\frac{1}{\noiseVar} \sum_{\Findex=0}^{\numF-1} \pilotObsEq^{*}[\RXindex,\Findex] \stMat(\UEposTest)[\RXindex,\Findex]}^2 \nonumber \\
        + & \sum_{\Findex=0}^{\numF-1}  \sum_{\Dindex=0}^{\numD-1} \log \left( \sum_{\constSymbol \in \constSet_{\constMap(\Findex,\Dindex)}} 
        \exp \left( \frac{2}{\noiseVar} \Re \left\{ \constSymbol\sum_{\RXindex=0}^{\numRX-1} \dataObs^{*}[\RXindex,\Findex,\Dindex] \channelConstruct(\UEposTest)[\RXindex,\Findex]  \right\} - \frac{1}{\noiseVar} \abs{\constSymbol}^2 \sum_{\RXindex=0}^{\numRX-1} \abs{\channelConstruct(\UEposTest)[\RXindex,\Findex] }^2 \right)
        \right)
        \label{eq:MMLa_estimator}
    \end{align}
    \endgroup
\end{figure*}

\begin{remark}\label{rem:MMLa}
    Similarly to $\UEposTest^{\MMLo}$, the contributions of different receivers combine incoherently in the pilot term,
    while in the data term they combine pseudo-coherently. 
    Indeed, the phase correction provided by $\channelcoeffEstPilots$ realigns the observations across the spatial dimension.
    However, since the data symbols differ across subcarriers and time instants, coherence in both frequency and time is lost in the second term.
    Intuitively, the second term of \eqref{eq:MMLa_estimator} sums over all possible transmitted symbols on each subcarrier $\Findex$ and time instant $\Dindex$. Ideally, the correct symbol contributes the most when $\UEposTest = \UEpos$, guiding the estimator toward the ground-truth value.
\end{remark}
\begin{remark}\label{rem:MMLa_comp}
  Finally, unlike \eqref{eq:MMLo_estimator}, which requires exhaustively enumerating all possible combinations $\dataMat \in \constSet_{\constMap}^{\numF \times \numD}$, the data term in \eqref{eq:MMLa_estimator} only requires evaluating all possible constellation points for each $(\Findex,\Dindex)$ pair \textbf{individually} ($s \in \constSet_{\constMap(\Findex,\Dindex)}$), resulting in a tractable computational complexity.  
\end{remark}

\subsection{Acceleration for Quadrature Amplitude Modulation}
The estimator $\UEposEst^{\MMLa}$ \eqref{eq:MMLa_estimator} incurs a computational complexity that is linear in $\constSizeMax$, denoted by $\Compl(\constSizeMax)$, as one must iterate over all constellation points for each time-frequency sample.
In \autoref{prop:MMLfast}, we show that when \gls{qam} constellations are used, the complexity can be reduced to $\Compl(\sqrt{\constSizeMax})$ while yielding the exact same solution, by analytically reformulating \eqref{eq:MMLa_estimator}.
For an $\constSize$-\gls{qam} constellation, the symbols are defined as \cite{3gpp_5g_2025}
\begin{equation}\label{eq:QAM_symbols}
    \constSymbol_{r,i} = \frac{1}{\sqrt{\constEnergy}} \left( \left(2r - \sqrt{\constSize} + 1\right) + \left(2i - \sqrt{\constSize} + 1\right) \jc \right),
\end{equation}
for $0 \leq r,i \in \mathds{Z} \leq \sqrt{\constSize} - 1$, and with $\constEnergy \triangleq 2(\constSize-1)/3$.
\begin{figure*}[t]
    \begingroup \small
    \begin{align}
       & \UEposEst^{\MMLfast} = \argmax_{\UEposTest}  \left( \frac{\pilotEnergy}{\noiseVar} + \frac{1}{\channelVar} \right)^{-1} 
        \sum_{\RXindex=0}^{\numRX-1} \abs{\frac{1}{\noiseVar} \sum_{\Findex=0}^{\numF-1} \pilotObsEq^{*}[\RXindex,\Findex]  \stMat(\UEposTest)[\RXindex,\Findex]}^2 \nonumber \\
        + & \sum_{\Findex=0}^{\numF-1} \sum_{\Dindex=0}^{\numD-1}\sum_{\mathfrak{F} \in \{ \Re, \Im \}} \log \left( \sum_{\substack{k=1 \\ k \text{ odd}}}^{\sqrt{\constSizeqd}-1}
         2 \cosh \left( \frac{2k}{\noiseVar \sqrt{\constEnergyqd}} \mathfrak{F} \left\{ \sum_{\RXindex=0}^{\numRX-1} \dataObs^{*}[\RXindex,\Findex,\Dindex] \channelConstruct(\UEposTest)[\RXindex,\Findex] \right\}  \right)
        \exp \left( - \frac{k^2}{\noiseVar\constEnergyqd}  \sum_{\RXindex=0}^{\numRX-1} \abs{\channelConstruct(\UEposTest)[\RXindex,\Findex]}^2 
        \right) \right)
        \label{eq:MMLfast_estimator}
    \end{align}
    \endgroup
\end{figure*}

\begin{proposition}\label{prop:MMLfast}
    By substituting the explicit expression of the \gls{qam} symbols into \eqref{eq:MMLa_estimator} and grouping terms of similar amplitude along both the real and imaginary axes, the estimator reduces to Equation \eqref{eq:MMLfast_estimator}, where $\constSizeqd = \#\constSet_{\constMap(\Findex,\Dindex)}$.
\end{proposition}
\begin{proof}
    See \hyperref[app:MMLfast_proof]{Appendix~\ref*{app:MMLfast_proof}}.
\end{proof}

\begin{remark}
    Rather than iterating over all constellation points, $\UEposEst^{\MMLfast}$ \eqref{eq:MMLfast_estimator} only requires two passes over $\sqrt{\constSizeqd}/2$ values, separately for the real and imaginary parts.
    This reduction is achieved by grouping constellation points of equal amplitude, as illustrated in \autoref{fig:qam}, which highlights that a \gls{qam} constellation can be decomposed into a set of \gls{bpsk}-like constellations at increasing amplitude levels along each axis.
    This structure is also exploited analytically in the proof provided in \hyperref[app:MMLfast_proof]{Appendix~\ref*{app:MMLfast_proof}}.
\end{remark}

\begin{figure}[ht]
    \centering
    \resizebox{0.65\linewidth}{!}{%
        \begin{tikzpicture}[scale=0.8]

    \draw[->, thick, myBlue] (-4,0) -- (4,0) node[right] {\Large $\Re$};
    \draw[->, thick, myBlue2] (0,-4) -- (0,4) node[above] {\Large $\Im$};

    \draw[dashed, gray!40] (-3,-3) grid (3,3);

    \foreach \x in {-3,-1,1,3}
        \foreach \y in {-3,-1,1,3}
            \filldraw[black!60] (\x,\y) circle (3pt);

    \draw[<->, ultra thick, dashed, myBlue]
        (-1,1) to[out=60, in=120] (1,1);

    \draw[<->, ultra thick, dashed, myBlue]
        (-3,1) to[out=60, in=120] (3,1);

    \draw[<->, ultra thick, dashed, myBlue2]
        (1,-1) to[out=30, in=-30] (1,1);

    \draw[<->, ultra thick, dashed, myBlue2]
        (1,-3) to[out=30, in=-30] (1,3);

\end{tikzpicture}
    }
    \caption{$16$-\gls{qam} constellation decomposed into \gls{bpsk}-like amplitude levels along the real and imaginary axes (indicated by dashed arrows).}
    \label{fig:qam}
\end{figure}
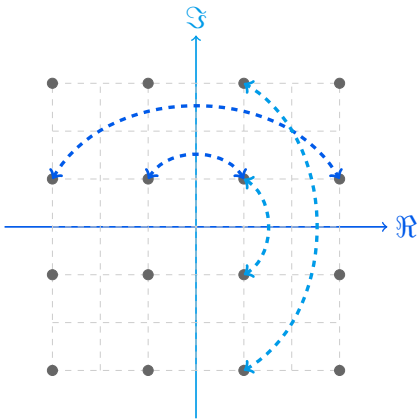

Finally, it is worth emphasizing that this reformulation is purely analytical:
$\MMLfast$ \eqref{eq:MMLfast_estimator} is guaranteed to yield \textbf{identical position estimates} to $\MMLa$ \eqref{eq:MMLa_estimator}, i.e., $\UEposEst^{\MMLfast} = \UEposEst^{\MMLa}$, while achieving significant computational savings, especially for large constellations.
Note that a detailed complexity analysis of both variants is presented in Section~\ref{sec:Complexity}.
\section{Numerical Results}\label{sec:NumericalResults}

This section presents \gls{mc} simulation results to assess the localization accuracy of the proposed estimator $\UEposEst^{\MMLfast}$ and compare it against existing methods from the literature. 

\subsection{Baselines}

For all methods, localization is performed in two steps: a \textbf{coarse estimate} obtained by searching over a discrete grid of $\Ngrid$ points defined over the scene, followed by a \textbf{refinement} using an optimization algorithm (here, the Nelder-Mead method \cite{nelder_simplex_1965}).
The following baselines are considered for comparison.

\subsubsection{\texorpdfstring{\textnormal{\Pmethod}}{Pilot}}
The simplest estimator relies solely on the pilot component, discarding the data part entirely.
This is the \textbf{traditional approach} employed in current systems, and is obtained by retaining only the first term of \eqref{eq:MMLfast_estimator} and discarding constant coefficients:
    \begin{align}
        \UEposEst^{\Pmethod} = \argmax_{\UEposTest} & 
        \sum_{\RXindex=0}^{\numRX-1} \Bigg|\underbrace{\sum_{\Findex=0}^{\numF-1} \pilotObsEq^{*}[\RXindex,\Findex] \stMat(\UEposTest)[\RXindex,\Findex]}_{\triangleq \corrP_{n}(\UEposTest)}\Bigg|^2,
        \label{eq:P_estimator}
    \end{align}
where $\corrP_{n}(\UEposTest)$ is introduced here for convenience, as it will be reused in the complexity analysis of Section~\ref{sec:Complexity}.

\subsubsection{\texorpdfstring{\textnormal{\PDmethod}}{Genie}}
This estimator represents the \textbf{best theoretically achievable performance}, obtained by assuming that the data symbols are perfectly known (or perfectly demodulated) at the \gls{srx}.
In this case, all $\numP+\numD$ symbols are treated as pilots by replacing the pilot-equalized observations $\pilotObsEq$ with the full-frame equivalent $\ObsEq \in \C^{\numRX \times \numF}$ in \eqref{eq:P_estimator}, defined as
\begin{equation}
    \ObsEq[\RXindex,\Findex] = \pilotObsEq[\RXindex,\Findex] + \sum_{\Dindex=0}^{\numD-1} \dataMat^{*}[\Findex,\Dindex] \dataObs[\RXindex,\Findex,\Dindex].
\end{equation}

\subsubsection{\texorpdfstring{\textnormal{\textsc{Decision-Directed}}}{Decision-Directed}}
In the \gls{dd} philosophy, the \gls{ue} position is estimated by using the \textbf{data symbol estimates as additional pilots}.
This method follows the expression of $\PDmethod$, replacing the known $\dataMat$ with its estimate, and thus suffers from potential demodulation errors.
Since the \gls{srx} is distributed, demodulation can be performed in either a \textit{distributed} or \textit{centralized} manner.
In the former, each of the $\numRX$ nodes independently estimates its $\numF$ channel coefficients from the pilot observations and performs data demodulation, which yields $\numRX$ potentially different sequences collected in $\dataTensEst \in \C^{\numRX \times \numF \times \numD}$.
In the latter, each node forwards its data observations to the \gls{cpu}, which estimates the channel coefficients and performs centralized data estimation to produce a single sequence $\dataMatEst \in \C^{\numF \times \numD}$.
Note that in each approach, both the channel and data estimations are based on the \gls{lmmse} criterion\footnote{A \gls{zf}-based estimation was also considered but provides results very similar to, though slightly worse than, the \gls{lmmse}-based approach and is therefore not shown here for graph readability.}.
The \textit{distributed} \gls{dd} estimator takes the form
\begingroup\small
\begin{align}
    & \UEposEst^{\mathrm{DD}} = \argmax_{\UEposTest} \sum_{\RXindex=0}^{\numRX-1} \Bigg| \underbrace{\sum_{\Findex=0}^{\numF-1} \ObsEqDD^{*}[\RXindex,\Findex] \stMat(\UEposTest)[\RXindex,\Findex]}_{\triangleq \corrDD_{\RXindex}(\UEposTest)}\Bigg|^2 , \, \text{with} \label{eq:DD_estimator} \\
    & \ObsEqDD[\RXindex,\Findex] \triangleq \pilotObsEq[\RXindex,\Findex] + \sum_{\Dindex=0}^{\numD-1} \dataTensEst^{*}[\RXindex,\Findex,\Dindex]  \dataObs[\RXindex,\Findex,\Dindex], \label{eq:DD_estimator_ObsEq}
\end{align}
\endgroup
while the \textit{centralized} variant is obtained by replacing $\dataTensEst[\RXindex,\Findex,\Dindex]$ with $\dataMatEst[\Findex,\Dindex]$ in \eqref{eq:DD_estimator_ObsEq}.
As before, $\corrDD_{\RXindex}(\UEposTest)$ is defined here for reference in Section~\ref{sec:Complexity}.
The data estimates can either be mapped to the closest constellation point, referred to as \textit{hard} decisions (\gls{hdd}), or left in the complex plane without mapping, referred to as \textit{soft} decisions (\gls{sdd}).
In the following, observations referring to \gls{dd} without further specification apply to both \gls{hdd} and \gls{sdd}.
Note that these \gls{dd} methods operate directly on symbols without accounting for error-correction coding, as we assume the passive system either lacks access to coding information or prefers to avoid performing such computationally intensive decoding.

\subsection{Simulation Framework and Metrics}
To evaluate localization performance, we report the \gls{rmse} of the position estimate,
\begin{equation}
\rmse\{\UEposEst\} = \sqrt{\E\{\norm{\UEposEst - \UEpos}^2\}},
\end{equation}
as a function of the per-node average \gls{snr}, defined as
\begin{equation}
    \snr = \frac{\E_{\RXindex, \channelcoeff}\left\{ \abs{\channelcoeff[\RXindex]}^2 \right\} \constVar}{\noiseVar},
\end{equation}
where $\constVar$ is the symbol variance, assumed identical for pilot and data symbols and normalized to $\constVar = 1$.
Under Assumptions \ref{ass:stillness} and \ref{ass:los}, no fading is considered.
Following the physical optics approximation \cite{kishk_high_2011}, the channel coefficients are modeled in simulation as
\begin{equation}
    \channel[\RXindex] = e^{\jc \channelPhase_{\RXindex}} 
    \frac{1}{\norm{\UEpos - \RXpos{\RXindex}}},
\end{equation}
where $\channelPhase_{\RXindex} \sim \U_{[0,2\pi)}$ accounts for the lack of phase synchronization across nodes.

\glsreset{ser}
Additionally, to highlight the interplay between sensing and communication performance, we consider the following communication metrics: the \gls{ser}, based on hard decisions, and the \gls{mae} of the soft data estimates.

While the developed framework is general and applies to any \gls{das} geometry, we adopt a \glsdesc{uca} of radius $\SRXradius$ and aperture $\SRXaperture \in (0, 2\pi]$ in the simulations.
The \gls{ue} is placed uniformly at random on a disk of radius\footnotemark\ $\Sradius < \SRXradius$.
\footnotetext{A small margin is introduced to avoid the singularity of the $\norm{\UEpos - \RXpos{\RXindex}}^{-1}$ path-loss model as $\norm{\UEpos - \RXpos{\RXindex}}\rightarrow0$.}
Each node experiences a different attenuation, while $\noiseVar$ is kept constant across nodes and set from the averaged \gls{snr}.
The averaged \gls{snr} is defined as $\snr = 2 / (\Sradius^2 \noiseVar)$, which follows from the rough approximation \begingroup\small$\E_{\RXindex, \channelcoeff}\{ \abs{\channelcoeff[\RXindex]}^2\} = \E_{\RXindex, \UEpos}\{ \norm{\UEpos - \RXpos{\RXindex}}^{-2} \} \approx 2/\Sradius^2$\endgroup \, in this configuration.

Unless otherwise stated, the simulation parameters are summarized in \autoref{tab:parameters}, where the frequency band used by the \gls{ue} follows future recommendations for \gls{isac} in \gls{6g} \cite{baduge_frequency_2025}.

\begin{table}[t]
    \centering
    \caption{Default Parameters used in Simulation}
    \label{tab:parameters}
        \begin{tabular}{|l|c|}
            \hline
            \multicolumn{2}{|c|}{\textbf{\gls{ue}}} \\ \hline \hline
            Carrier frequency $\carrierF$ & $\SI{7.2}{\giga\hertz}$ \\
            Subcarrier spacing $\Fspacing$ & $\SI{60}{\kilo\hertz}$ \\ 
            Number of subcarriers $\numF$ & $128$ \\
            Number of pilot \gls{ofdm} symbols $\numP$ & $1$ \\
            Number of data \gls{ofdm} symbols $\numD$ & $35$ \\
            Pilots & \gls{bpsk} \\
            Data constellation $\constSet_{\constMap(\Findex,\Dindex)}$ & $256$-\gls{qam} ($\forall \Findex,\Dindex$)\\ 
            \hline \hline 
            \multicolumn{2}{|c|}{\textbf{\gls{srx}}} \\ \hline \hline 
            Number of nodes $\numRX$ & $5$ \\
            Radius $\SRXradius$ & $\SI{5000}{\carrierWl} = \SI{208.3}{\meter}$ \\
            Aperture $\SRXaperture$ & $2\pi~\si{\radian}$ \\
            \hline \hline 
            \multicolumn{2}{|c|}{\textbf{Simulation}} \\ \hline \hline 
            Scene radius $\Sradius$ & $\SI{4800}{\carrierWl} = \SI{200}{\meter}$ \\
            \makecell[l]{Number of grid points $\Ngrid$ \\ \begingroup\scriptsize(square area of side $2\Sradius$)\endgroup} & $40^2=1600$ \\
            Refinement optimization method & $\texttt{Nelder-Mead}$ \cite{nelder_simplex_1965} \\
            Number of \gls{mc} iterations $\Nmc$ & $3000$ \\
            \hline 
        \end{tabular}
\end{table}

\vspace*{-0.2cm}
\subsection{Localization Performance}
\begin{figure*}[t]
    \includegraphics[width=1.0\textwidth]{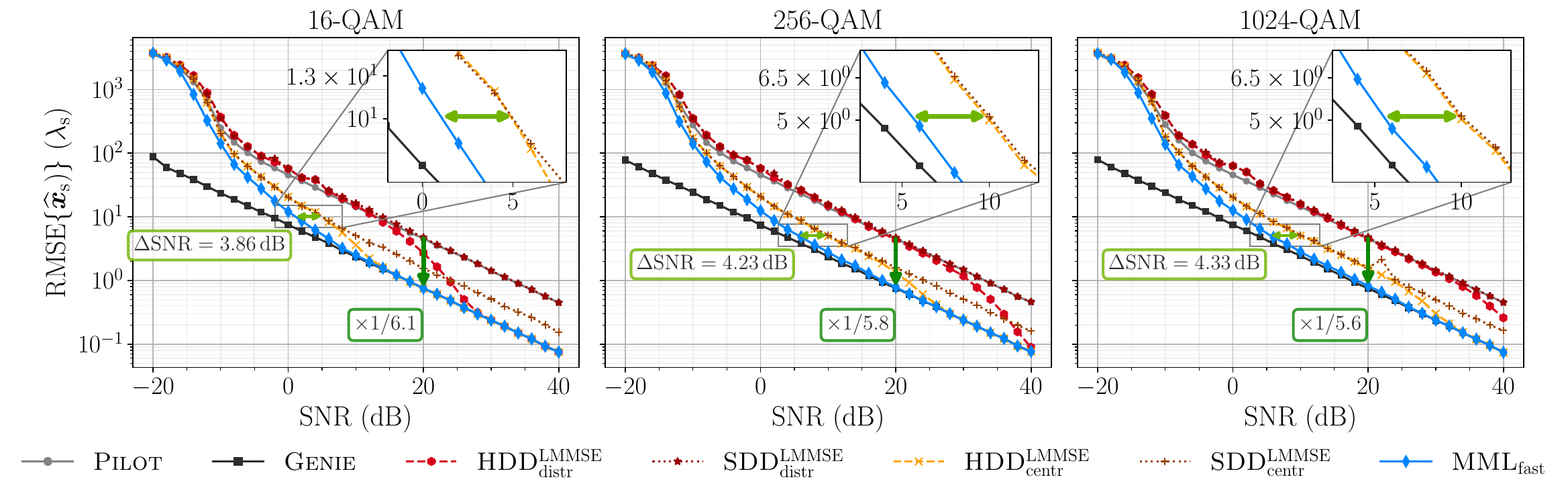}
    \caption{\gls{rmse} as a function of \gls{snr} for different data constellations. \textit{Hard} (resp. \textit{soft}) \gls{dd} baselines are represented by dashed (resp. dotted) lines. The proposed method is shown as a solid blue line with diamond markers, while $\Pmethod$ and $\PDmethod$ are shown as solid lines with circle and square markers, respectively. Note that the \gls{rmse} is normalized with respect to $\carrierWl$ (here $\SI{4.17}{\centi\meter}$).}
    \label{fig:RMSE_constellations}
\end{figure*}
\begin{figure*}[t]
    \centering
    \begin{minipage}{0.5\textwidth}
        \centering
        \includegraphics[width=0.8\linewidth]{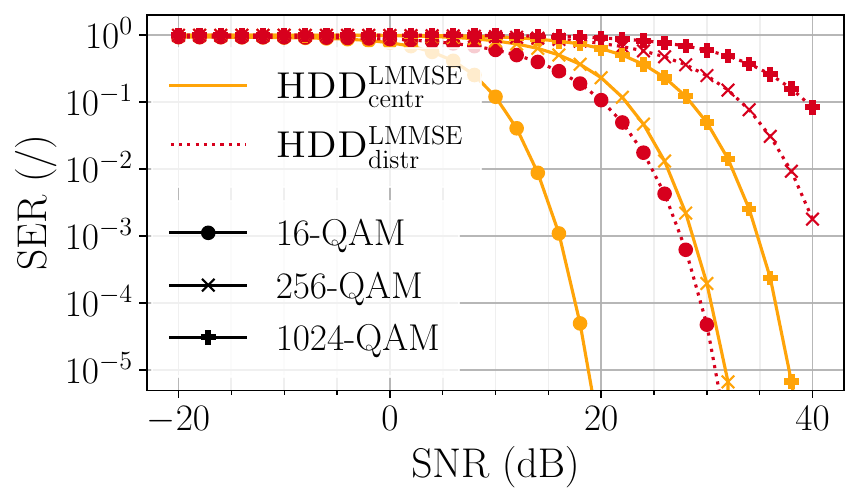}
        \caption{\gls{ser} as a function of the \gls{snr}, for different data constellations.}
        \label{fig:SER_constellations}
    \end{minipage}%
    \begin{minipage}{0.5\textwidth}
        \centering
        \includegraphics[width=0.8\linewidth]{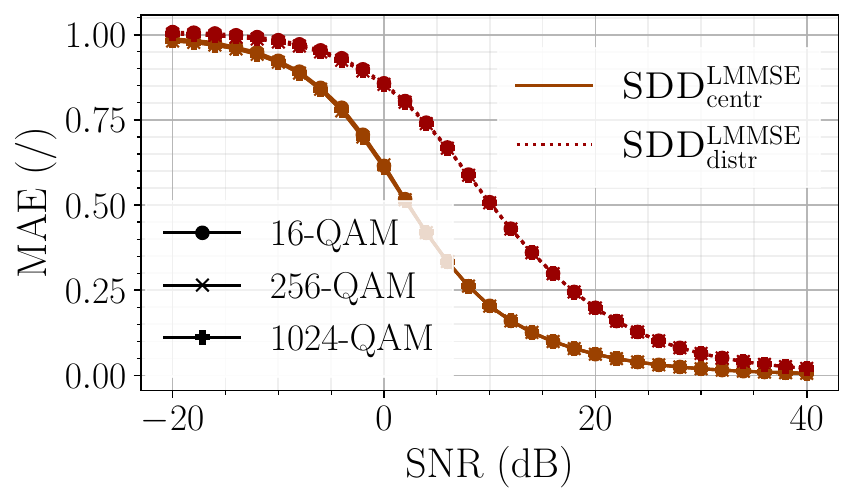}
        \caption{\gls{mae} as a function of the \gls{snr}, for different data constellations.}
        \label{fig:MAE_constellations}
    \end{minipage}
\end{figure*}

Figures \ref{fig:RMSE_constellations} and \ref{fig:SER_constellations} show the localization \gls{rmse} and communication \gls{ser}, respectively, as a function of \gls{snr} for different constellation sizes.
The following observations can be drawn:
\begin{itemize}
    \item As expected, the \textit{distributed} \gls{dd} methods perform worse than the \textit{centralized} ones, as collaborative demodulation improves the quality of the data symbol estimates (\autoref{fig:SER_constellations}) and thus the localization performance (\autoref{fig:RMSE_constellations}), while the potentially different sequences $\dataTensEst$ across nodes in the distributed case degrade the coherence being created in \eqref{eq:DD_estimator_ObsEq}. 

    \item Once all data symbols are correctly demodulated, the \gls{hdd} methods converge to the $\PDmethod$ performance. 
    In practice, this convergence is observed when $\ser \lesssim 10^{-2}$.
    As expected, this occurs at higher \gls{snr} for larger constellations, due to increased demodulation difficulty. 
     
    \item The proposed $\MMLfast$ approach achieves \textbf{significant localization improvement over the traditional pilot-only} method, with a reduction in \gls{rmse} by a factor of up to 6 across all constellations and under the considered parameter settings (highlighted by \textcolor{myGreenArrow2}{$\bm{\downarrow}$}).\\
    
    \item The proposed $\MMLfast$ approach \textbf{consistently converges to the $\PDmethod$ bound at lower \glspl{snr} than any \gls{dd} method}, regardless of the constellation. 
    Additionally, while \gls{hdd} methods suffer from degraded localization performance as the constellation size grows, the proposed approach remains \textbf{robust to constellation size}: the increased symbol uncertainty does not result in a noticeable increase in localization error, which is particularly beneficial for high data-rate communications.
    For the considered set of parameters, \gls{snr} improvements of up to \SI{4.3}{\decibel} over the best \gls{dd} baseline are observed (highlighted by \textcolor{myGreenArrow1}{$\bm{\leftrightarrow}$}). Note that this robustness comes at the computational cost of iterating over a larger number of constellation points since $\constSizeqd$ increases in \eqref{eq:MMLfast_estimator}, as will be discussed in Section~\ref{sec:Complexity}.
\end{itemize}

\vspace*{-0.2cm}
\subsection{Impact of Errors on DD Localization}\label{sec:impact_data_errors}

It is also instructive to analyze the effect of symbol errors on the behavior of \gls{dd} baselines.
Comparing the \gls{rmse}  against the \gls{ser}, we observe that for the \textit{centralized} \gls{hdd}, \textbf{demodulation errors are less detrimental to localization performance than one might expect}\footnotemark.
For instance, for $1024$-\gls{qam} and $\snr \in [0, 20]$~\si{\decibel}, exploiting the full frame with data decisions significantly improves localization compared to using pilots only (\autoref{fig:RMSE_constellations}), even when the \gls{ser} remains high (\autoref{fig:SER_constellations}), confirming that positioning information is present even when symbols are not correctly demodulated. 
This can be attributed to the discrete nature of the \gls{ser} metric: a soft estimate mapped to an incorrect constellation point is counted as a full symbol error, regardless of how close it lies to the true symbol. 
This is further corroborated by \autoref{fig:MAE_constellations}, which shows the \gls{mae} of the soft data symbol estimates for the same constellation sizes. 
The \gls{mae} exhibits a significant decrease over $\snr \in [0, 20]$~\si{\decibel} and is, as expected, independent of the constellation size.
Consequently, the \textbf{\gls{mae}} (soft estimation) \textbf{governs the first \gls{rmse} drop} of \textnormal{$\HDDcentrmethod$} over \textnormal{$\Pmethod$}—consistent across all constellation sizes—while the \textbf{\gls{ser}} (constellation mapping) \textbf{governs the second drop} toward \textnormal{$\PDmethod$} performance, which occurs at higher \gls{snr} for larger constellations.
\footnotetext{A similar observation was reported by Henninger~et~al.\ \cite{henninger_hybrid_2026} in a different scenario, for a single-antenna receiver producing a single demodulated sequence—analogous to the \textit{centralized} case considered here—further supporting the finding that sensing from data symbols does not require correct decoding of the entire frame.}
This is further confirmed by the \gls{rmse} of the \textnormal{$\SDDcentrmethod$} method (\autoref{fig:RMSE_constellations}), which benefits only from the first gain\footnotemark.
This effect is especially pronounced for large constellations (i.e., a larger \gls{snr} gap between the two drops), where a \textbf{misclassified symbol} may still lie geometrically \textbf{close to the correct one} in the complex plane, and thus \textbf{carry substantial positioning information} that contributes beneficially to localization.

Finally, \gls{hdd} methods only close their performance gap with \textnormal{$\PDmethod$} when data estimation becomes nearly perfect, which occurs at increasingly high \gls{snr} for larger constellations.
\textbf{In contrast, the proposed \gls{mml} approach closes this gap at significantly lower \gls{snr}, consistently across all constellation sizes.}

\footnotetext{As observed in \autoref{fig:RMSE_constellations}, the \textit{distributed} \gls{hdd} benefits only from the second gain, while the \textit{distributed} \gls{sdd} yields no improvement over $\Pmethod$ across the considered \gls{snr} range.
This behavior is due to the coherence degradation induced by the potentially differing sequences, as previously discussed.}
\subsection{Impact of System Parameters}\label{sec:results_impact_parameters}

This section analyzes the impact of system parameters on the localization performance. 
A summary of the results is provided in \autoref{tab:performance_analysis}.
\begin{table}[ht]
    \centering
    \caption{Impact of system parameters on localization \gls{rmse}.}
    \label{tab:performance_analysis}
        \begin{tabular}{|l|c|c|c|c|c|}
            \hline
             & $\Pmethod$ & $\PDmethod$ & $\mathrm{HDD}$ & $\MMLfast$ & Figure \\
            \hline \hline 
            $\uparrow \numP$ & $\downarrow$ &  $\downarrow$ &  $\downarrow$ &  $\downarrow$ & -  \\ \hline
            $\uparrow \numD$ & $-$ &  $\downarrow$ &  $\downarrow$ &  $\downarrow$ & \ref{fig:Sweep_D} \\ \hline 
            $\uparrow \constSizeqd$ & $-$ &  $-$ &  $\uparrow$ &  $-$ & \ref{fig:RMSE_constellations} \\ \hline
            $\uparrow \numRX$ & $\downarrow \downarrow$ &  $\downarrow \downarrow$ &  $\downarrow \downarrow$ &  $\downarrow \downarrow$ & \ref{fig:Sweep_N} \\\hline 
            $\uparrow \numF$ & $\downarrow \downarrow$ &  $\downarrow \downarrow$ &  $\downarrow \downarrow$ &  $\downarrow \downarrow$  & \ref{fig:Sweep_Q}\\ 
            \hline 
        \end{tabular}
\end{table}

\subsubsection{Number of data symbols}
\autoref{fig:Sweep_D} shows the \gls{rmse} as a function of $\numD$ for different \gls{snr} values. 
Increasing $\numD$ naturally improves the performance of the $\PDmethod$ bound and benefits \gls{hdd} methods, although these gains become significant only at \gls{snr} values where data estimation is nearly perfect. 
Similarly, $\MMLfast$ benefits from additional data observations only when the \gls{snr} is sufficiently high for the data symbols to meaningfully guide the estimator toward the true solution (see Remark~\ref{rem:MMLa}). 
At lower \gls{snr}, increasing $\numD$ yields only limited performance gains, as reflected by the saturation behavior observed in the figure.
Unlike pilot symbols, which directly reduce the \gls{rmse} by increasing the useful signal energy—so that increasing $\numP$ consistently improves estimation (as observed for $\PDmethod$ here)—increasing $\numD$ enhances the performance of $\MMLfast$ and \gls{hdd} methods only up to a certain point, beyond which a plateau is reached due to imperfect data estimation.
This indicates that the available positioning information becomes limited by the data uncertainty, such that further increases in $\numD$ provide diminishing returns.
Finally, it is worth noting that this regime corresponds to very large $\numD$, for which \autoref{ass:stillness} may no longer hold.
\begin{figure}[ht]
    \centering
    \includegraphics[width=1.0\columnwidth]{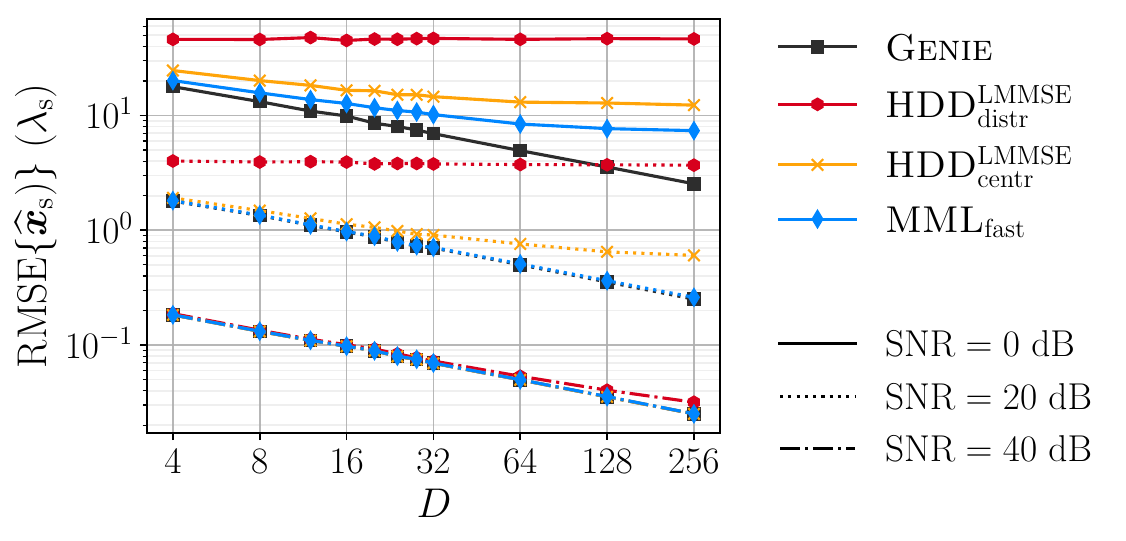}
    \caption{\gls{rmse} as a function of $\numD$ for different \gls{snr} values. Parameters: $\numRX=8$, $\numF=64$, $\Fspacing=\SI{120}{\kilo\hertz}$ (yielding the same bandwidth as in \autoref{tab:parameters}), $\Nmc=2000$; all other parameters are left at their default values.}
    \label{fig:Sweep_D}
\end{figure}

\subsubsection{Number of RX nodes}
\autoref{fig:Sweep_N} depicts the \gls{rmse} with respect to $\numRX$ for multiple \gls{snr} values.
Increasing $\numRX$ reduces the \gls{rmse} for all methods. 
For $\MMLfast$, this improves both the array (i.e., non-coherent) gain ($\sum_{\RXindex}$ outside the $\abs{\cdot}^2$) of the pilot term and the \textit{pseudo}-coherent combining of the data observations in \eqref{eq:MMLfast_estimator}.
For centralized \gls{hdd} approaches, additional nodes lead to better estimation of the data sequence—whose impact on localization has been discussed in Section~\ref{sec:impact_data_errors}—while both centralized and distributed \gls{hdd} methods benefit from the increased array gain over the combined pilot symbols (i.e., pilots and estimated data symbols treated jointly as $\numP + \numD$ pilots, as seen in \eqref{eq:DD_estimator_ObsEq}).

\begin{figure}[ht]
    \centering
    \includegraphics[width=1.0\columnwidth]{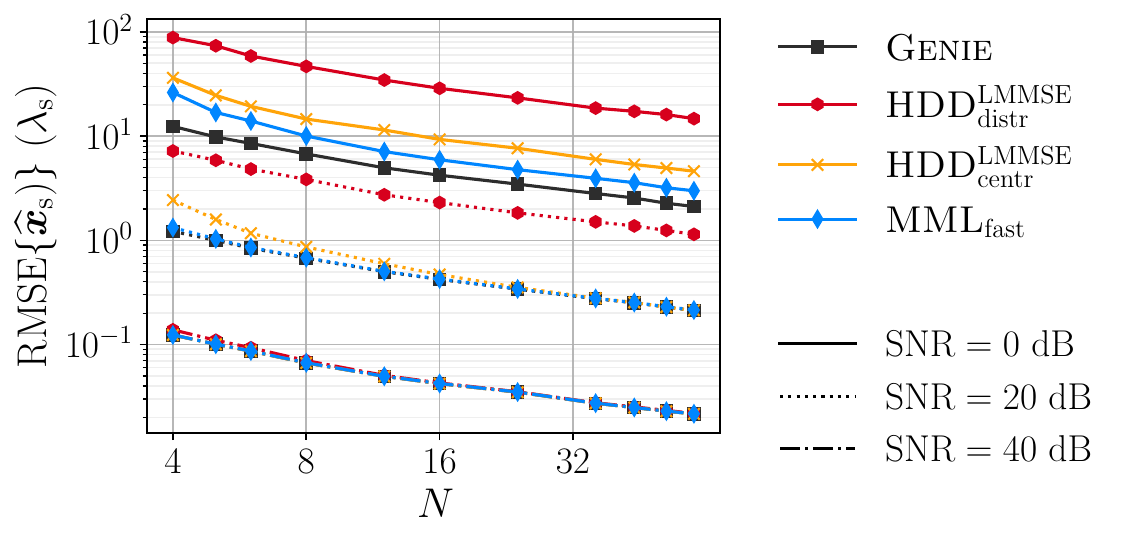}
    \caption{\gls{rmse} as a function of $\numRX$ for different \gls{snr} values. Parameters: $\numF=64$, $\Fspacing=\SI{120}{\kilo\hertz}$ (yielding the same bandwidth as in \autoref{tab:parameters}), $\Nmc=2000$; all other parameters are left at their default values.}
    \label{fig:Sweep_N}
\end{figure}

\subsubsection{Number of subcarriers}

\autoref{fig:Sweep_Q} represents the \gls{rmse} as a function of $\numF$ for different \glspl{snr}. 
Increasing $\numF$ improves localization for all methods, though the performance gap between methods again depends on the operating \gls{snr}, i.e., whether the data observations provide near-perfect positioning information or not. 
This improvement benefits all methods through the bandwidth coherent gain ($\sum_{\Findex}$ inside the $\abs{\cdot}^2$) in the pilot term—or combined pilot term for \gls{hdd} methods—while $\MMLfast$ additionally benefits from the non-coherent bandwidth gain in the data term of \eqref{eq:MMLfast_estimator}.
Hence, it is observed that at low \gls{snr}, increasing $\numF$ benefits the proposed approach more strongly than the \gls{hdd} baselines.
\begin{figure}[ht]
    \centering
    \includegraphics[width=1.0\columnwidth]{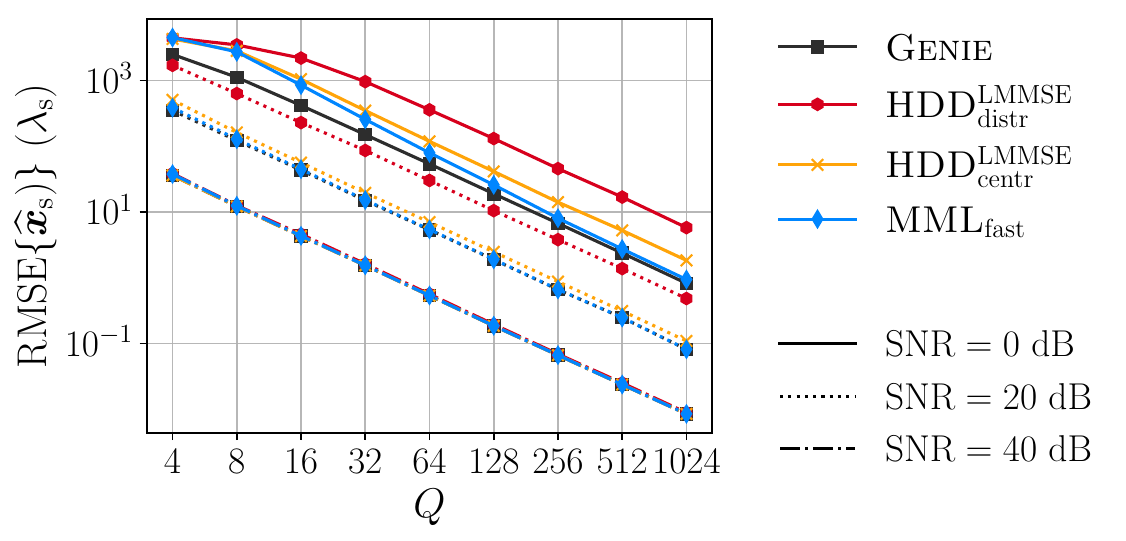}
    \caption{\gls{rmse} as a function of $\numF$ for different \gls{snr} values. Parameters: $\numRX=8$, $\Fspacing=\SI{120}{\kilo\hertz}$, $\Nmc=2000$; all other parameters are left at their default values.}
    \label{fig:Sweep_Q}
\end{figure}

The effect of the bandwidth coherent gain is illustrated in \autoref{fig:AF_Q}, which shows the \gls{af} for different values of $\numF$ for both the space-coherent and non-space-coherent cases, and highlights the intrinsic range resolution $\rangeRes = \frac{c}{2\BW} = \frac{\carrierF \carrierWl}{\numF \Fspacing}$ (indicated by the double-headed arrows in the figure).
These \gls{af} functions are defined as
\begingroup \small
\begin{align}
    &\AF_{\mathrm{coh}}(\UEpos, \UEposTest) = \sum_{\RXindex=0}^{\numRX-1} \sum_{\Findex=0}^{\numF-1} e^{-\jc \carrierWn \left( \norm{\UEpos - \RXpos{\RXindex}} - \norm{\UEposTest - \RXpos{\RXindex}} \right) \Findex \frac{\Fspacing}{\carrierF}}, \\
    & \AF_{\mathrm{non-coh}}(\UEpos, \UEposTest) = \sum_{\RXindex=0}^{\numRX-1} \abs{\sum_{\Findex=0}^{\numF-1} e^{-\jc \carrierWn \left( \norm{\UEpos - \RXpos{\RXindex}} - \norm{\UEposTest - \RXpos{\RXindex}} \right) \Findex \frac{\Fspacing}{\carrierF}}}^2,
\end{align}
\endgroup
where the attenuation term is omitted since the localization procedure does not exploit this information.
These functions correspond to the noise-free output of a filter matched to the candidate position $\UEposTest$ for a source located at $\UEpos$ \cite{monnoyer_chirp-based_2025-1}, following a \gls{ml} approach in the presence of \gls{awgn}.
In particular, the non-coherent \gls{af} directly corresponds to the noise-free log-likelihood of the pilot-only estimator \eqref{eq:P_estimator} when $\numP=1$, and therefore serves as a useful tool to assess the achievable estimation performance given the structural dependence of the observations on the \gls{poi} \cite{monnoyer_chirp-based_2025-1}, here $\UEpos$. 
When $\numF=1$, localization is infeasible as no coherence dimension remains.
Increasing $\numF$ sharpens the main lobe and thus improves range resolution (i.e., reduces $\rangeRes$), thereby enhancing localization performance\footnotemark.
\newcounter{sharedfn}
\setcounter{sharedfn}{\value{footnote}}.

The figure also illustrates the information loss incurred by non-coherent processing, which arises from the lack of phase synchronization across the distributed array nodes.
Furthermore, note that increasing the subcarrier spacing $\Fspacing$ can further improve this resolution\footnotemark[\value{sharedfn}] (i.e., reduce $\rangeRes$) without requiring additional observations along $\Findex$. 
\begin{figure}[ht]
    \centering
    \includegraphics[width=0.94\columnwidth]{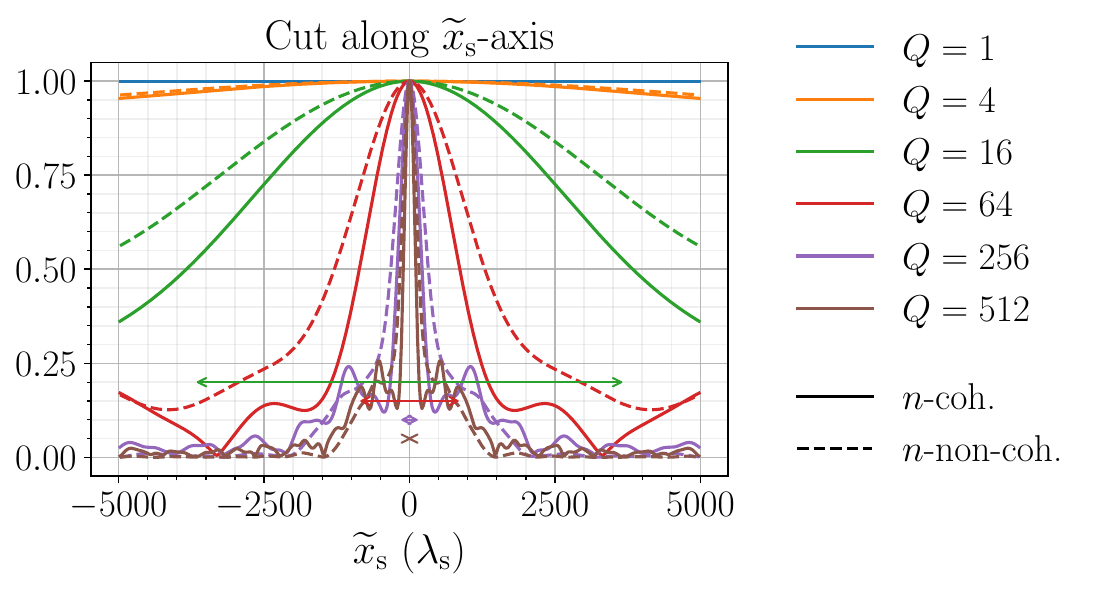}
    \caption{Normalized coherent and non-coherent \gls{af}, $\abs{\AF_{\mathrm{coh}}(\boldsymbol{0}, \UEposTest)}$ and $\abs{\AF_{\mathrm{non-coh}}(\boldsymbol{0}, \UEposTest)}$, for different values of $\numF$. A one-dimensional cut along $\widetilde{x}_{\mathrm{s}} = \UEposTest[0]$ is shown. Parameters: $\Ngrid = 400^2$; all other parameters are left at their default values.}
    \label{fig:AF_Q}
\end{figure}

\footnotetext[\value{sharedfn}]{A sufficiently large $\Ngrid$ was used in the simulations to ensure that the coarse grid captures the main lobe of the objective function. 
This is enforced by setting $\Ngrid = \max(\Ngrid, \alpha \frac{2\Sradius}{\rangeRes})$, where $\alpha \in \Z_{+}$ is a predefined oversampling factor, set to $\alpha=4$ in simulation.}

Finally, increasing $\numRX$ (\autoref{fig:Sweep_N}) or $\numF$ (\autoref{fig:Sweep_Q}) yields stronger localization improvement than increasing $\numD$ (\autoref{fig:Sweep_D}), since additional nodes and subcarriers enrich the steering matrix $\stMat(\UEpos) \in \C^{\numRX \times \numF}$ in both the spatial and frequency domains, whereas increasing $\numD$ (or $\numP$) only provides additional observations without adding new structural information.


\section{Complexity Analysis}\label{sec:Complexity}

This section analyzes the computational complexity of the proposed method against the considered baselines in terms of asymptotic operation counts, illustrating the influence of system parameters on computational requirements.
Note that we also highlight operations that can be performed in parallel (i.e., locally at each node rather than at the \gls{cpu}), and discuss the associated transmission cost, i.e., which quantities need to be forwarded to the \gls{cpu}.

\subsection{Theoretical Analysis}\label{sec:Complexity_th}
\autoref{tab:complexity} reports the asymptotic complexity of the processing steps for each method, and for an $\constSize$-ary constellation, i.e., $\forall \Findex,\Dindex \, \#\constSet_{\constMap(\Findex,\Dindex)} = \constSize$.
The first two columns correspond to the proposed methods, \eqref{eq:MMLa_estimator} and \eqref{eq:MMLfast_estimator}, without and with the acceleration, respectively.
The asymptotic complexity of the localization step is expressed by two terms corresponding to the pilot and data contributions, respectively, each of which must be evaluated across all $\Ngrid$ grid points\footnotemark.
\footnotetext{
After grid evaluation, estimates are refined via an optimization algorithm requiring additional objective function evaluations.
Since the number of such evaluations depends on the solver strategy rather than on the proposed estimators, it is excluded from the theoretical complexity analysis.
In practice, a sufficiently fine search grid yields adequate localization accuracy without refinement.
}
Note that the symbol equalization step is independent of the candidate position and needs to be performed only once prior to the localization step.
The data contribution is in $\Compl(\numRX + \sqrt{\constSize})$ (resp. $\Compl(\numRX + \constSize)$) for $\MMLfast$ (resp. $\MMLa$) rather than $\Compl(\numRX \sqrt{\constSize})$ (resp. $\Compl(\numRX \constSize)$), since the summation over $\RXindex$ can be performed independently of the symbol index $k$ (resp. symbol $s$).
Furthermore, the data term additionally requires the computation of $\channelcoeffEstPilots(\UEposTest)$ \eqref{eq:channelcoeffEstPilots} for all $\UEposTest$, at a cost of $\Compl(\Ngrid\numRX\numF)$, which merges with the pilot term in the asymptotic complexity.
The remaining columns correspond to the baselines introduced above.
The \gls{sdd} approaches share the same complexity as their \gls{hdd} counterparts, without the hard data decision step.
\begin{table*}[ht]
    \centering
    \caption{Asymptotic Computational Complexity per Processing Step}
    \label{tab:complexity}
    \resizebox{\textwidth}{!}{
    \begingroup 
    \setlength{\tabcolsep}{4pt} 
    \renewcommand{\arraystretch}{1.5}
    \begin{tabular}{|l||c|c|c|c|c|c|}
        \hline
        Step & \textcolor{MMLacolor}{$\bm{\MMLa}$} \eqref{eq:MMLa_estimator} & \textcolor{MMLfastcolor}{$\bm{\MMLfast}$} \eqref{eq:MMLfast_estimator} & \textcolor{Pcolor}{\textbf{\Pmethod}} & \textcolor{PDcolor}{\textbf{\PDmethod}} & \textcolor{DDcentrcolor}{$\bm{\HDDcentrmethod}$} & \textcolor{DDdistrcolor}{$\bm{\HDDdistrmethod}$} \\
        \hline 
        \makecell[l]{Channel estimation} & — & — & — & — & $\Compl(\numRX\numF\numP)$ & $\Compl(\numRX\numF\numP)$ \\ 
        \hline
        \makecell[l]{Soft data estimation} & — & — & — & — & $\Compl(\numRX\numF\numD)$ & $\Compl(\numRX\numF\numD)$  \\ 
        \hline 
        \makecell[l]{Hard data decision} & — & — & — & — & $\Compl(\constSize\numF\numD)$ & $\Compl(\numRX\constSize\numF\numD)$ \\
        \hline
        \makecell[l]{Symbol Equalization i.e.,  \\
        construction of $\pilotObsEq$, $\ObsEq$ or $\ObsEqDD$} & $\Compl(\numRX\numF\numP)$ & $\Compl(\numRX\numF\numP)$  & $\Compl(\numRX\numF\numP)$  & $\Compl(\numRX\numF(\numP+\numD))$ & $\Compl(\numRX\numF(\numP+\numD))$  & $\Compl(\numRX\numF(\numP+\numD))$  \\
        \hline
        \makecell[l]{Localization} & \makecell[c]{$\Compl(\Ngrid \numRX\numF)$ \\ $+ \Compl(\Ngrid \numF\numD (\numRX + \constSize))$} & \makecell[c]{$\Compl(\Ngrid \numRX\numF)$ \\ $+ \Compl(\Ngrid \numF\numD (\numRX + \sqrt{\constSize}))$} & 
        $\Compl(\Ngrid \numRX\numF)$ & $\Compl(\Ngrid \numRX\numF)$  & $\Compl(\Ngrid \numRX\numF)$  & $\Compl(\Ngrid \numRX\numF)$   \\
        \hline
    \end{tabular}
    \endgroup
    }
\end{table*}

The following subsections respectively examine acceleration gains, comparison with baselines, and transmission overhead.

\subsubsection{Gain of the Acceleration}
As already mentioned, the analytical acceleration $\MMLfast$ \eqref{eq:MMLfast_estimator} reduces the complexity of the data term from $\Compl(\constSize)$ in $\MMLa$ \eqref{eq:MMLa_estimator} to $\Compl(\sqrt{\constSize})$, while yielding the exact same solution, i.e., $\UEposEst^{\MMLa} = \UEposEst^{\MMLfast}$.
Therefore, $\MMLfast$ should always be preferred over $\MMLa$.
Note that this reformulation is therefore particularly beneficial for high-data-rate communications using large constellations.

\subsubsection{\texorpdfstring{$\MMLfast$}{MMLfast} against baselines}\label{sec:Complexity_th_MMLfast_vs_baselines}
The cheapest method is the pilot-only approach, as exploiting data payloads for positioning incurs an unavoidable computational cost.
Regarding the localization step, the proposed $\MMLfast$ incurs an additional cost of $\Compl(\Ngrid \numF\numD (\numRX +  \sqrt{\constSize}))$ over $\Pmethod$,
while \gls{hdd} methods share the same localization cost as $\Pmethod$ but incur computational overhead from the symbol equalization step, at a cost of $\Compl(\numRX\numF\numD)$.
Additionally, the \gls{dd} methods require a one-time data demodulation overhead: $\Compl(\numRX\numF\numD)$ for soft data estimation, and $\Compl(\constSize\numF\numD)$ and $\Compl(\numRX\constSize\numF\numD)$ for centralized and distributed hard decisions, respectively.
Note that data demodulation is performed only once in \gls{hdd} methods, whereas our approach iterates over the constellation for every position candidate.
Consequently, $\MMLfast$ is cheaper than $\HDDcentrmethod$ only when $\Ngrid (\numRX + \sqrt{\constSize}) \ll \numRX + \constSize$, and cheaper than $\HDDdistrmethod$ only when $\Ngrid (\numRX + \sqrt{\constSize}) \ll \numRX \constSize$.
The former condition is seldom met in practice, while the latter may be satisfied for small grid sizes, large constellations, and a sufficient number of nodes.
Hence, the localization performance gains of the proposed approach over \gls{dd} baselines are generally achieved at the expense of an increased computational cost.

\subsubsection{Transmission Overhead}
In the pilot-only approach \eqref{eq:P_estimator}, pilot symbols $\pilotMat$ are known at each node, allowing node $\RXindex$ to locally compute $\CorrP_{\RXindex} \triangleq \left\{ \corrP_{n}(\UEposTest)\right\} \in \R^{\Ngrid}$, i.e., the set of $\corrP_{n}(\UEposTest)$ values evaluated over all candidate grid positions.
Each node then transmits $\CorrP_{\RXindex}$ to the \gls{cpu}, which estimates the \gls{ue} position as $\argmax_{\UEposTest} \sum_{\RXindex} \abs{\corrP_{n}(\UEposTest)}^2$ (see \eqref{eq:P_estimator}).
This results in a transmission of $\Ngrid$ scalars per node.
By contrast, $\MMLfast$ and $\DDcentrmethod$ (similarly for the soft and hard variants) both require each node to additionally transmit all data observations $\dataObs$ to the \gls{cpu}—to evaluate the second term of \eqref{eq:MMLfast_estimator} and to collaboratively estimate $\dataMatEst$, respectively—resulting in a transmission of $\numF\numD$ coefficients per node for the data part.
The distributed \gls{dd} approach reduces this transmission cost: each node $\RXindex$ locally produces its own data estimates $\dataTensEst[\RXindex,\Findex,\Dindex] \, \forall \Findex, \Dindex$, computes $\CorrDD_{\RXindex} \triangleq \left\{\corrDD_{\RXindex}(\UEposTest)\right\} \in \R^{\Ngrid}$ (see \eqref{eq:DD_estimator}), and transmits only $\CorrDD_{\RXindex}$ to the \gls{cpu}, i.e., $\Ngrid$ scalars per node.
For the considered parameter set, $\Ngrid=1600 < \numF\numD=4480$, yielding a moderate but tangible reduction in transmission overhead compared to $\MMLfast$ and $\DDcentrmethod$.
This transmission reduction is the only advantage of $\DDdistrmethod$ in this scenario, as its localization performance is inferior, as discussed above.

\section{Conclusion}\label{sec:Conclusion}

This paper investigates the joint exploitation of pilot and data payloads for source localization where a distributed opportunistic \gls{srx} aims to localize the transmitting \gls{ue}. 
We develop an \gls{mml} framework in which both the random channel coefficients at all nodes and the unknown data symbols are treated as \glspl{np}, eliminated through marginalization. 
The optimal solution is derived and shown to be computationally intractable.
We then introduce a practical approximation by eliminating the channel dependency in the data term, leveraging pilot-based channel estimates.
We further propose a novel closed-form acceleration for \gls{qam} modulations, reducing the computational complexity from the full constellation size in $\Compl(\constSize)$ to its square root $\Compl(\sqrt{\constSize})$. 
Through Monte Carlo simulations, the proposed approach demonstrates superior localization performance against benchmarks from the literature, namely the traditional pilot-only baseline and several \gls{dd} approaches. 
It achieves a significant \gls{rmse} reduction of up to a factor of 6 over the former, and an \gls{snr} gain of up to \SI{4.3}{\decibel} over the best \gls{dd} baseline, for the considered parameter settings.
Two key findings further emerge: unlike \gls{hdd} methods, whose performance degrades with increasing modulation order, the proposed approach remains robust to constellation size and consistently achieves convergence to the $\PDmethod$ bound at lower \glspl{snr}, regardless of the constellation.
These results show that marginalizing over the constellation systematically outperforms data demodulation for localization, regardless of the modulation order.
This performance improvement comes at the cost of an unavoidable increased computational complexity relative to \gls{dd} baselines, as discussed in the complexity analysis.
Additionally, we study the impact of data demodulation errors on \gls{dd}-based sensing by jointly considering localization and communication metrics, showing that even imperfect demodulation can still yield localization gains over pilot-only methods.
Future work may consider multipath propagation, multistatic radar configurations relying on signal echoes, and extensions to multi-user scenarios, as well as imperfect synchronization and experimental validations.

\appendix
\subsection{Proof of \autoref{prop:MMLo}}\label{app:MMLo_proof}

\begin{proof}
    Marginalizing over $\channelcoeff$ yields
    \begin{align}
         \Likelihood(\Obs | \dataMat; \UEposTest) & = \ExpChannel{\Likelihood(\Obs | \dataMat, \channelcoeff; \UEposTest)} \\
        & = \int \Likelihood(\Obs | \dataMat, \channelcoeff; \UEposTest) \Likelihood(\channelcoeff) d\channelcoeff. \label{eq:MMLo_integration_channel}
    \end{align}
    From observation models \eqref{eq:pilotObs} and \eqref{eq:dataObs}, the elements of $\Obs$ are \gls{iid} and follow $\Obs[\RXindex,\Findex,\PDindex] \sim \CN(\channelcoeff[\RXindex] \stMat(\UEpos)[\RXindex,\Findex]\symbolMat[\Findex,\PDindex],\noiseVar)$. 
    Additionally, we assume $\channelcoeff \sim \CN(\boldsymbol{0}, \channelVar \I{\numRX})$.
    Therefore, the integration in \eqref{eq:MMLo_integration_channel} results in
    \begingroup \small
    \begin{align}
        & \Likelihood(\Obs | \dataMat; \UEposTest) = \frac{1}{\pi^{K+\numRX} \noiseVarK \channelVarN}
        \int \prod_{\RXindex} 
        \exp\left( \frac{-1}{\channelVar} \abs{\channelcoeff[n]}^2\right)
        \nonumber \\
        & \quad 
        \exp\left( 
            -\frac{1}{\noiseVar}
            \sum_{\Findex,\PDindex} \abs{\Obs[\RXindex,\Findex,\PDindex] - \channelcoeff[\RXindex] \stMat(\UEposTest)[\RXindex,\Findex]\symbolMat[\Findex,\PDindex]}^2
        \right) d \channelcoeff, 
    \end{align}
    \endgroup
    where $K = \numRX \numF \numPD $ is the total number of observations.
    Since the observations and the channel coefficients are independent across receivers, the integral factorizes as 
    \begin{equation}\label{eq:property_integrate_iid}
        \int f(\boldsymbol{t}) d\boldsymbol{t} = \int \prod_n f(t_n) d\boldsymbol{t} = \prod_n \int f(t_n) dt_n.
    \end{equation} 
    Expanding the squared term and moving factors independent of $\channelcoeff$ outside the integral yields
    \begingroup \small
    \begin{align}
        & \Likelihood(\Obs | \dataMat; \UEposTest) = \frac{1}{\pi^{K+\numRX} \noiseVarK \channelVarN}
        \prod_{\RXindex} \exp \left( \frac{-1}{\noiseVar} \sum_{\Findex,\PDindex} \abs{\Obs[\RXindex, \Findex,\PDindex]}^2 \right) \nonumber \\
        &  
        \int_{\channelcoeffn} \exp\left(2\Re\left\{ U_{\RXindex} \channelcoeffn\right\}  - V \abs{\channelcoeffn}^2 \right) d \channelcoeffn, \label{eq:MMLo_expression_sub}
    \end{align}
    \endgroup
    where
    \begin{align}
        U_{\RXindex} & \triangleq \frac{1}{\noiseVar} \sum_{\Findex, \PDindex} \Obs^{*}[\RXindex, \Findex, \PDindex] \stMat(\UEposTest)[\RXindex, \Findex] \symbolMat[\Findex, \PDindex], \\
        V & \triangleq \frac{1}{ \noiseVar} \sum_{\Findex,\PDindex} \abs{\symbolMat[\Findex,\PDindex]}^2 + \frac{1}{\channelVar},
    \end{align}
    as $\abs{\stMat(\UEposTest)[\RXindex,\Findex]}^2 = 1$.
    Completing the square in the exponent, the remaining integral evaluates to
    \begingroup \small
    \begin{align}
        & \int_{\channelcoeffn} \exp \left(2\Re\left\{ U_{\RXindex} \channelcoeffn\right\}  - V \abs{\channelcoeffn}^2  \right) d \channelcoeffn \nonumber \\
        &  = \int_{\channelcoeffn} \exp \left(-V \abs{\frac{U_{\RXindex}^{*}}{V} - \channelcoeffn}^2 + \frac{\abs{U_{\RXindex}}^2}{V}  \right) d \channelcoeffn \\
        & = \exp \left(\frac{\abs{U_{\RXindex}}^2}{V} \right) \frac{\pi}{V} \label{eq:MMLo_integral_result},
    \end{align}
    \endgroup
    where we used the property $\int_{z} e^{-\alpha \abs{z}^2} dz = \pi / \alpha$.
    Substituting \eqref{eq:MMLo_integral_result} along with the expressions of $U_{\RXindex}$ and $V$ back into \eqref{eq:MMLo_expression_sub} leads to
    \begin{align}
        & \Likelihood(\Obs | \dataMat; \UEposTest) = \textcolor{gray}{\frac{1}{\pi^{K+\numRX} \noiseVarK \channelVarN}} \nonumber\\
        & \quad \prod_{\RXindex} \textcolor{gray}{\exp \left( \frac{-1}{\noiseVar} \sum_{\Findex, \PDindex} \abs{\Obs[\RXindex,\Findex,\PDindex]}^2 \right)}  \frac{\textcolor{gray}{\pi}}{\frac{1}{\noiseVar} \sum_{\Findex,\PDindex} \abs{\symbolMat[\Findex, \PDindex]}^2 + \frac{1}{\channelVar}} 
        \nonumber \\
        & \quad \times \exp \left( \frac{\abs{
            \frac{1}{\noiseVar} \sum_{\Findex,\PDindex} \Obs^{*}[\RXindex, \Findex,\PDindex] \stMat(\UEposTest)[\RXindex,\Findex]\symbolMat[\Findex,\PDindex]}^2
        }{\frac{1}{\noiseVar} \sum_{\Findex,\PDindex} \abs{\symbolMat[\Findex, \PDindex]}^2 + \frac{1}{\channelVar}} \right),
    \end{align}
    where the \textcolor{gray}{gray terms} are independent of $\UEposTest$ and $\symbolMat$ and can thus be absorbed into a multiplicative constant $C$, given by
    \begin{equation}
        C = \frac{\exp \left(-\frob{{\Obs}}^2 / \noiseVar \right) }{\pi^{K+\numRX-1} \noiseVarK \channelVarN}.
    \end{equation}
    It remains to compute $\Likelihood(\pilotObs, \dataObs; \UEposTest) = \ExpData{\Likelihood(\pilotObs, \dataObs | \dataMat; \UEposTest)}$.
    Since the summations over $\Findex$ and $\PDindex$ appear inside the squared norm (i.e., frequency and time coherence are preserved), the sums cannot be factorized.
    Consequently, one must exhaustively enumerate all possible symbol combinations $\dataMat \in \constSet_{\constMap}^{\numF \times \numD}$, yielding
    \begin{equation}
       \Likelihood(\pilotObs, \dataObs; \UEposTest) = \sum_{\dataMat \in \constSet_{\constMap}^{\numF \times \numD}} 
        \Likelihood(\Obs | \dataMat; \UEposTest) \Likelihood(\dataMat),
    \end{equation}
    which concludes the proof.
\end{proof}

\subsection{Proof of \autoref{prop:MMLa}}\label{app:MMLa_proof}

\begin{proof}
The pilot term follows directly from \autoref{prop:MMLo} by restricting to the known pilot sequence $\pilotMat$ rather than the entire frame $\symbolMat$, discarding constant terms, denoting $\pilotEnergy \triangleq \frob{\pilotMat}^{2}$ and substituting \eqref{eq:pilotObsEq}.
Assuming a uniform distribution over the constellations, the data term is computed as
\begin{align}
    & \ExpData{ \Likelihood(\dataObs | \dataMat; \UEposTest, \channelcoeffEstPilots(\UEposTest))} \nonumber \\
    & = 
    \int_{\dataMat} \prod_{\RXindex,\Findex,\Dindex} \Likelihood(\dataObs[\RXindex, \Findex, \Dindex] | \dataMat; \UEposTest, \channelcoeffEstPilots(\UEposTest)) \Likelihood(\dataMat) d \dataMat \label{eq:MMLa_derivation_s1} \\
    & = \prod_{\Findex,\Dindex} \sum_{\constSymbol \in \constSet_{\constMap(\Findex,\Dindex)}} \prod_{\RXindex} \Likelihood(\dataObs[\RXindex, \Findex, \Dindex] | \constSymbol ; \UEposTest, \channelcoeffEstPilots(\UEposTest)) \frac{1}{\# \constSet_{\constMap(\Findex,\Dindex)}},  \label{eq:MMLa_derivation_s2}
\end{align}
where \eqref{eq:MMLa_derivation_s1} leverages the independence of the observations, and \eqref{eq:MMLa_derivation_s2} exploits the property given in \eqref{eq:property_integrate_iid}, accounting for the discrete nature of the symbols. 
The likelihood in \eqref{eq:MMLa_derivation_s2} is given by
\begin{align}
    & \Likelihood(\dataObs[\RXindex, \Findex, \Dindex] | \constSymbol ; \UEposTest, \channelcoeffEstPilots(\UEposTest)) \nonumber \\
    & = \frac{1}{\pi \noiseVar} \exp \left(\frac{-1}{\noiseVar} \abs{\dataObs[\RXindex,\Findex,\Dindex]-\channelcoeffEstPilots(\UEposTest)[\RXindex]\stMat(\UEposTest)[\RXindex,\Findex] \constSymbol}^2 \right). \label{eq:dataObs_likelihood}
\end{align}
Substituting \eqref{eq:dataObs_likelihood} into \eqref{eq:MMLa_derivation_s2}, expanding the squared norm, taking the logarithm of \eqref{eq:MMLa_formulation}, and omitting constant terms yields \eqref{eq:MMLa_estimator}, which completes the proof.
\end{proof}

\subsection{Proof of \autoref{prop:MMLfast}}\label{app:MMLfast_proof}
\begin{proof}
Let us define \begingroup\small$\mathrm{H}_{\Findex}(\UEposTest) \triangleq \sum_{\RXindex=0}^{\numRX-1} |\channelConstruct(\UEposTest)[\RXindex,\Findex]|^2$\endgroup \,
and \begingroup\small$S_{\Findex\Dindex}^{*}(\UEposTest) \triangleq \sum_{\RXindex=0}^{\numRX-1} \dataObs^{*}[\RXindex,\Findex,\Dindex] \channelConstruct(\UEposTest)[\RXindex,\Findex]$\endgroup.
For a $\constSizeqd$-\gls{qam} constellation, the data term likelihood (inside the $\log$) in \eqref{eq:MMLa_estimator} can be rewritten as

\begin{align}
    & \mathsf{S} \triangleq \sum_{\constSymbol \in \constSet_{\constMap(\Findex,\Dindex)}} 
        \exp \left( \frac{2}{\noiseVar}  \Re \left\{S_{\Findex\Dindex}^{*}(\UEposTest) \constSymbol  \right\} - \frac{1}{\noiseVar} \mathrm{H}_{\Findex}(\UEposTest) \abs{\constSymbol}^2 \right) \nonumber \\
    & = \sum_{r,i}
    \exp \left( \frac{1}{\noiseVar} \left( 2 \Re \left\{ S_{\Findex\Dindex}^{*}(\UEposTest) \constSymbol_{r,i} \right\} - \mathrm{H}_{\Findex}(\UEposTest)\abs{\constSymbol_{r,i}}^2 \right) \right), \label{eq:MMLfast_derivation_s1}
\end{align}
where $\constSymbol_{r,i}$ is given by \eqref{eq:QAM_symbols} and $0 \leq r,i \in \mathds{Z} \leq \sqrt{\constSizeqd} - 1$.
The remainder of the proof proceeds as follows. First, \eqref{eq:QAM_symbols} is substituted into \eqref{eq:MMLfast_derivation_s1}, and the real part and squared magnitude of the products involving $\constSymbol_{r,i}$ with $S_{\Findex\Dindex}^{*}(\UEposTest)$ and $\mathrm{H}_{\Findex}(\UEposTest)$ are expanded using the identities $\Re\{z_1 z_2 \} = \Re \{ z_1 \} \Re \{ z_2 \} - \Im \{ z_1 \} \Im \{ z_2 \}$ and $\abs{z_1}^2 =  (\Re \{ z_1 \})^2 +   (\Im \{ z_1 \})^2$, $\forall z_1, z_2 \in \C$. 
The real and imaginary contributions are then separated, yielding a product of a sum over $r$ and a sum over $i$. 
Finally, the summation indices are reparametrized as $r' = 2r - \sqrt{\constSize} + 1$, with $r'$ ranging over odd integers in $[-\sqrt{\constSize}+1, \sqrt{\constSize}-1]$ (and similarly for $i'$).
These steps yield the following intermediate result:
\begin{align}
     \mathsf{S} = \prod_{\mathfrak{F}\in \{\Re, \Im \}}  \Bigg[
         \sum_{\substack{k=-\sqrt{\constSize}+1 \\ k \text{ odd}}}^{\sqrt{\constSize}-1}  & \exp\left(  \frac{1}{\noiseVar} 2 \mathfrak{F} \left\{S_{\Findex\Dindex}^{*}(\UEposTest) \right\} \frac{k}{\sqrt{\constEnergyqd}} \right) \nonumber \\
    &  \exp\left( -\frac{1}{\noiseVar} \mathrm{H}_{\Findex} \frac{k^2}{\constEnergyqd} \right)\Bigg],
\end{align}
where $k$ is the running index replacing $r'$ and $i'$.
The symmetric pairs $(-k, k)$, corresponding to constellation points of equal amplitude (see \autoref{fig:qam}), are then grouped together, enabling the use of the hyperbolic cosine identity.
Taking the logarithm of the resulting expression completes the proof.
\end{proof}

\begingroup
\linespread{0.98}\selectfont  
\bibliographystyle{IEEEtran}
\bibliography{nourl, references}
\endgroup

\end{document}